\documentclass[a4paper,twocolumn,english,showpacs,nofootinbib,nobibnotes,aps,prl,floatfix,superscriptaddress]{revtex4-2}
\usepackage[utf8]{inputenc}
\usepackage[T1]{fontenc}
\usepackage{amsmath}
\usepackage{amsthm}
\usepackage{thmtools}

\usepackage{amssymb}
\usepackage{graphicx}
\usepackage[table]{xcolor}
\usepackage{rotating}
\usepackage{mathtools}
\usepackage{fancyhdr}
\usepackage{nicefrac}

\usepackage[caption=false]{subfig}
\usepackage[colorlinks,citecolor=blue,linkcolor=orange,urlcolor=blue]{hyperref}
\usepackage{adjustbox}

\usepackage{xparse}

\NewDocumentCommand\opti{smmmm>{\SplitList{;}}m} {
\begingroup%
\setlength{\belowdisplayskip}{-0.6\baselineskip}%
\IfBooleanTF{#1}{%
    \begin{alignat*}{2}
        #2& \underset{#4}{\text{#3}} & & #5 \\
          & \text{s.~t.~~}
        \ProcessList{#6}{ \insertopticonst }
          & &
    \end{alignat*}%
    }{%
    \begin{alignat}{2}
        #2& \underset{#4}{\text{#3}} & & #5 \\
        & \text{s.~t.~~}
        \ProcessList{#6}{ \insertopticonst }
        & & \nonumber
    \end{alignat}%
    }%
\endgroup%
}%
\newcommand\insertopticonst[1]{& & #1\\&}

\usepackage{tikz}
\usepackage{tikz-cd}
\usepackage{pgfplots}
\usetikzlibrary{plotmarks}
\pgfplotsset{compat=newest}

\newcommand{\tikzmark}[1]{\tikz[overlay,remember picture] \node (#1) {};}
\tikzset{square arrow/.style={to path={-- ++(0,-.25) -| (\tikztotarget)}}}
\tikzset{deep square arrow/.style={to path={-- ++(0,-.5) -| (\tikztotarget)}}}

\usetikzlibrary{shapes.misc}
\usetikzlibrary{arrows.meta}
\newcommand{\startgrid}[2]{%
 \begin{tikzpicture}[darkstyle/.style={draw,fill=black!67,minimum size=28}, lightstyle/.style={draw,fill=black!0,minimum size=28}, nostyle/.style=   {fill,fill=white, line width=0.0mm,minimum size=28}, arrowstyle/.style={line width=0.5mm,-{Latex[length=3mm,width=5mm]}},noarrowstyle/.style={line width=0.5mm},  baseline={(0,0)}]
  \foreach \x in {1,...,#1}
    \foreach \y in {1,...,#2}
       \node [lightstyle] (\y\x) at (\x-1,-\y+1) { $\bullet$ }; 
}

\newcommand{\stopgrid}{%
\end{tikzpicture}
}

\newcommand{\etatwotwo}{%
    \startgrid{2}{2}
    
    \draw[arrowstyle] (0,-0)--(1,-1) ;
    
    \stopgrid
}

\newcommand{\etathreetwo}{%
    \startgrid{3}{3}
    
    \draw[arrowstyle] (0,-0)--(1,-1);
    \draw[arrowstyle] (0,-1)--(1,-2);
    \draw[arrowstyle] (1,-0)--(2,-1);
    
    \stopgrid
}

\newcommand{\etafourtwo}{%
    \startgrid{4}{4}
    
    \draw[arrowstyle] (0,-0)--(1,-1);
    \draw[arrowstyle] (0,-1)--(1,-2);
    \draw[arrowstyle] (0,-2)--(1,-3);
    \draw[arrowstyle] (1,-0)--(2,-1);
    \draw[arrowstyle] (2,-0)--(3,-1);
    
    \draw[arrowstyle] (2,-2)--(3,-3);
    
    \stopgrid
}

\newcommand{\etafivetwo}{%
    \startgrid{5}{5}
    
    \draw[arrowstyle] (0,-0)--(1,-1);
    \draw[arrowstyle] (0,-1)--(1,-2);
    \draw[arrowstyle] (0,-2)--(1,-3);
    \draw[arrowstyle] (0,-3)--(1,-4);
    \draw[arrowstyle] (1,-0)--(2,-1);
    \draw[arrowstyle] (2,-0)--(3,-1);
    \draw[arrowstyle] (3,-0)--(4,-1);
    
    \draw[arrowstyle] (2,-2)--(3,-3);
    \draw[arrowstyle] (2,-3)--(3,-4);
    \draw[arrowstyle] (3,-2)--(4,-3);
    
    \stopgrid
}

\newcommand{\etathreethree}{%
    \startgrid{3}{3}
    
    \draw[arrowstyle] (0,-0)--(1,-1)--(2,-2);
    
    \draw[noarrowstyle] (0,-1)--(1,-2);
    \draw[noarrowstyle,dotted] (1,-2)--(1.5,-2.5);
    \draw[arrowstyle,dotted] (1.5,0.5)--(2,-0);
    
    \draw[noarrowstyle] (1,-0)--(2,-1);
    \draw[noarrowstyle,dashed] (2,-1)--(2.5,-1.5);
    \draw[arrowstyle,dashed] (-0.5,-1.5)--(0,-2);
    
    \stopgrid
}

\newcommand{\etafivefive}{%
    \startgrid{5}{5}
    
    \draw[arrowstyle] (0,-0)--(1,-1)--(2,-2)--(3,-3)--(4,-4);
    
    \draw[noarrowstyle] (0,-1)--(1,-2)--(2,-3)--(3,-4);
    \draw[noarrowstyle,dotted] (3,-4)--(3.5,-4.5);
    \draw[arrowstyle,dotted] (3.5,0.5)--(4,-0);
    
    \draw[noarrowstyle] (1,-0)--(2,-1)--(3,-2)--(4,-3);
    \draw[noarrowstyle,dashed] (4,-3)--(4.5,-3.5);
    \draw[arrowstyle,dashed] (-0.5,-3.5)--(0,-4);

    \draw[noarrowstyle] (0,-2)--(1,-3)--(2,-4);
    \draw[noarrowstyle,densely dotted] (2,-4)--(2.5,-4.5);
    \draw[noarrowstyle,densely dotted] (2.5,0.5)--(3,-0);
    \draw[arrowstyle] (3,-0)--(4,-1);

    \draw[noarrowstyle] (2,-0)--(3,-1)--(4,-2);
    \draw[noarrowstyle,loosely dotted] (4,-2)--(4.5,-2.5);
    \draw[noarrowstyle,loosely dotted] (-0.5,-2.5)--(0,-3);
    \draw[arrowstyle] (0,-3)--(1,-4);
    
    \stopgrid
}

\makeatletter

\newcounter{thmc}
\newcounter{thmcbak}

\declaretheorem[thmbox=S]{result}

\newtheorem{theorem}[thmc]{Theorem}
\newtheorem{corollary}[thmc]{Corollary}
\newtheorem{lemma}[thmc]{Lemma}

\newcommand{\pushright}[1]{\ifmeasuring@#1\else\omit\hfill$\displaystyle#1$\fi\ignorespaces}

\usepackage{braket}
\usepackage{dsfont}

\usepackage[normalem]{ulem}

\makeatother

\usepackage{babel}

\global\long\def\trace{\operatorname{Tr}}
\global\long\def\CZ{\operatorname{CZ}}
\global\long\def\NN{\boldsymbol{N}}
\global\long\def\ketu#1#2{\left\vert \substack{#1\\#2} \right\rangle}

\global\long\def\brau#1#2{\left\langle \substack{#1\\#2} \right\vert }

\global\long\def\ketbra#1{\ket{#1}\!\bra{#1}}
\global\long\def\ketbrau#1#2#3#4{\left\vert \substack{#1\\#2} \right\rangle\!\!\left\langle \substack{#3\\#4} \right\vert }

\global\long\def\ketbraa#1#2{\ket{#1}\!\bra{#2}}

\global\long\def\one{\mathds{1}}

\global\long\def\dout{d_\text{out}}
\global\long\def\din{d_\text{in}}
\global\long\def\GHZ{\text{GHZ}}
\global\long\def\LOSR{\text{LOSR}}

\newcommand{\EE}{\mathcal{E}}
\begin{document}

\title{
No quantum advantage without classical communication: \\
fundamental limitations of quantum networks}

\author{Justus Neumann}

\author{Tulja Varun Kondra}

\affiliation{Institut für Theoretische Physik~III, Heinrich-Heine-Universität Düsseldorf, Universitätsstr.~1, 40225 Düsseldorf, Germany}

\author{Kiara Hansenne}
\affiliation{Naturwissenschaftlich-Technische Fakultät, Universität Siegen, Walter-Flex-Str.~3, 57068 Siegen, Germany}
\affiliation{Université Paris-Saclay, CEA, CNRS, Institut de Physique Théorique, 91191 Gif-sur-Yvette, France}

\author{ Lisa~T. Weinbrenner}

\affiliation{Naturwissenschaftlich-Technische Fakultät, Universität Siegen, Walter-Flex-Str.~3, 57068 Siegen, Germany}

\author{Hermann Kampermann}

\affiliation{Institut für Theoretische Physik~III, Heinrich-Heine-Universität Düsseldorf, Universitätsstr.~1, 40225 Düsseldorf, Germany}

\author{Otfried Gühne}

\affiliation{Naturwissenschaftlich-Technische Fakultät, Universität Siegen, Walter-Flex-Str.~3, 57068 Siegen, Germany}

\author{Dagmar Bruß}

\author{Nikolai Wyderka}

\affiliation{Institut für Theoretische Physik~III, Heinrich-Heine-Universität Düsseldorf, Universitätsstr.~1, 40225 Düsseldorf, Germany}

\date{\today}

\begin{abstract}
Quantum networks connect systems at separate locations via quantum 
links, enabling a wide range of quantum information tasks between 
distant parties. Large-scale networks have the potential to enable 
global secure communication, distributed quantum computation, 
enhanced clock synchronization, and high-precision multiparameter 
metrology. For the optimal development of these technologies, however, 
it is essential to identify the necessary resources and sub-routines 
that will lead to the quantum advantage, but this is demanding even for the simplest protocols in quantum information processing.
Here we show that quantum networks relying on the long-distance 
distribution of bipartite entanglement, combined with local 
operations and shared randomness, cannot achieve a relevant quantum 
advantage. 
Specifically, we prove that these networks do not help in preparing 
resourceful quantum states such as Greenberger-Horne-Zeilinger states 
or cluster states, despite the {free} availability of {long-distance} 
entanglement. At an abstract level, our work points {towards} a 
fundamental difference between bipartite and multipartite entanglement. 
From a practical perspective, our results highlight the need for classical 
communication combined with quantum memories to fully harness the power of 
quantum networks.
\end{abstract}

  \maketitle

\section{Introduction}

\begin{figure*}[t]
    \centering
    \includegraphics[width=0.9\linewidth]{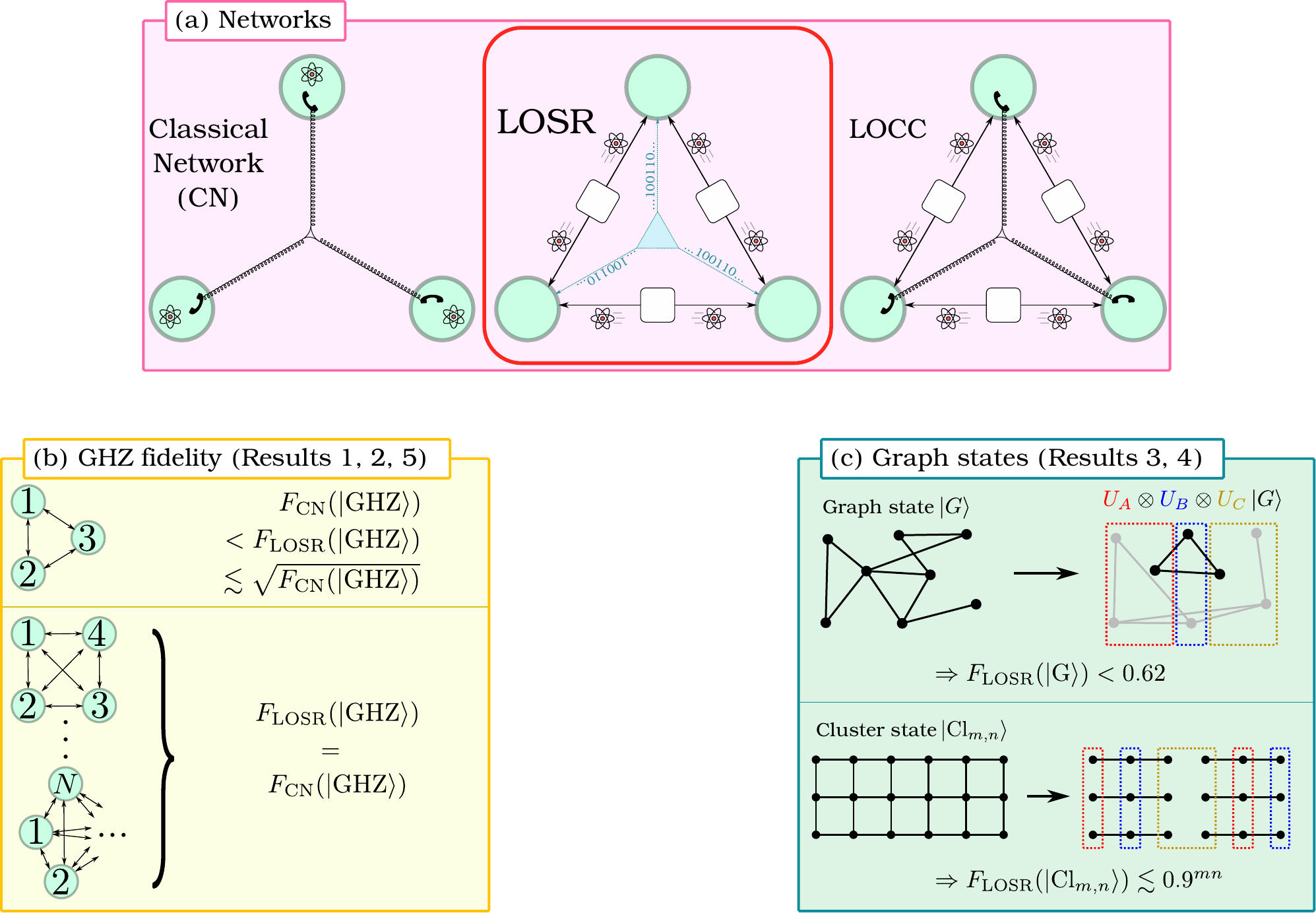}
    \caption{
    Summary of the setting and our results. (a) We consider quantum LOSR (local operations and shared randomness, center) networks, where distant parties are connected by bipartite entangled sources and have access to a source of shared randomness. We show that this scenario is, in terms of capabilities, much more similar to a completely classical network (CN, left) than to a quantum LOCC network (local operations and classical communications, right), where entangled sources are accompanied by classical communication, which allows for preparation of arbitrary quantum states. (b) In particular, we show that for the task of multipartite GHZ state preparation, a small advantage of LOSR networks over the classical network scenario is present only in tripartite networks, whereas it offers no advantage at all in four- and larger-partite networks. (c) More generally, we show that no connected graph state can be prepared with high fidelity in LOSR networks. Some relevant families of states such as cluster states, being  resource states for measurement-based quantum computation, even have exponentially decreasing fidelity bounds in terms of the number of parties.
    }

    \vspace{-0.3cm}
    \label{fig:overview}
\end{figure*}

Quantum networks promise to provide platforms for long-distance 
quantum information processing, enabling global secure communication 
and distributed quantum computation \cite{Kimble2008,Wehner2018, azuma2023quantum}. 
Consequently, many experimental efforts aim towards an experimental
realization of building blocks of network structures \cite{hermans2022qubit, liu2023multinode, hartung2024quantum}. A critical requirement for many  
applications is the ability of a network to generate and distribute 
entangled states among many distant parties. Particularly relevant 
among these states are quantum graph states, such as the 
Greenberger-Horne-Zeilinger (GHZ) states and cluster states.
Indeed, these states can then be used for applications like 
blind quantum computation \cite{Barz2012}, quantum metrology \cite{Sekatski2020,Proctor2018,Guo2020}, quantum key distribution protocols \cite{Murta2020} and clock synchronization \cite{Komar2014}.
Preparing and distributing such states in practice, however, presents significant challenges, as it requires precise control over 
high-dimensional quantum systems, a task that becomes increasingly 
complex with system size. 

In general, distributed information processing can be modeled in different 
ways, see also Fig.~\ref{fig:overview}(a). In the simplest scenario (Fig.~\ref{fig:overview}(a) left), one considers different parties 
controlling quantum mechanical systems in their labs, where the parties can communicate classically. This scenario, 
however, does not allow to generate entanglement between the parties,
so no quantum advantage can be expected. In the most advanced scenario (Fig.~\ref{fig:overview}(a) right), 
the parties not only communicate and control their systems, they also
share initial bipartite entanglement. In this scenario, protocols like teleportation
between the nodes can be carried out, allowing the preparation of arbitrary
global states as well as the execution of general information processing 
protocols. This, however, requires the experimental implementation of 
quantum memories, as well as performing quantum operations in one lab
in dependence of the measurement results communicated by the other 
parties.

In order to take into account the current experimental constraints, the 
community has studied intermediate models of networks, which are closer 
to the actual realization. One central model are networks based on local operations and shared randomness (LOSR), where parties are connected via 
bipartite quantum sources and can only rely on local operations along 
with pre-shared classical correlations 
\cite{buscemi2012all,Schmid2023understanding,deVicente2014,forster2009distilling,gutoski2009properties}. LOSR networks 
also start from distributed bipartite entanglement, 
but they do not rely on time-consuming classical communication that necessitates 
functioning quantum memories.  These characteristics make LOSR networks experimentally accessible, and as a result their properties and limitations have been actively studied in recent 
years \cite{navascues2020genuine,luo2021new,kraft2021quantum, hansenne2022symmetries, wang2024quantum, makuta2023no,yang2024quantumenhanced}.

But are quantum networks in the LOSR scenario really useful tools to 
unleash the full power of distributed quantum information processing? 
The main result of this paper is a negative answer to this question. 
More precisely, we consider the task of generating multipartite
highly entangled quantum states. We demonstrate that the maximal fidelity achievable by LOSR networks with a large class of resource states does 
not significantly exceed the values reachable by classical networks.  
In other words, the advantage provided by bipartite sources for preparing 
multipartite entangled states to be used by distant parties is minimal, and additional resources --- such as quantum memories, classical communication, 
or direct distribution of multipartite entangled states --- are essential. Our discoveries resemble an observation made in Ref. \cite{hansenne2022symmetries} where it was shown 
that for the special case of output states in the symmetric subspace only,  LOSR
networks cannot produce anything else than fully separable states.

Consequently, this limitation should not be seen as a merely negative result. 
In fact, no-go theorems have often served as starting points for major advancements, by forcing the research community to look for different 
routes. For example, the well-known constraints of linear optics in 
performing Bell state measurements motivated the development of 
alternative approaches, such as hyperentanglement
\cite{sheng2010complete,wei2007hyperentangled,walborn2003hyperentanglement,schuck2006complete} or the usage of ancilla photons combined with
photon-number resolving detectors \cite{ewert2014efficient,bayerbach2023bellstate}. Similarly, early recognized limitations in entanglement 
distillation with Gaussian states \cite{fiurasek2002gaussian,eisert2002distilling} ultimately led to the study of non-Gaussian quantum operations, culminating in the experimental implementation of entanglement distillation 
with non-gaussian operations \cite{takahashi2010entanglement}. Similarly, our results highlight the need to explore alternative network architectures, such as hybrid models incorporating classical communication together with quantum channels and memories,
to fully utilize the potential of quantum networks.

\section{Preparing multi-party quantum states}
If quantum states shall be prepared in a network, the 
$N$-qubit GHZ state
\begin{align}
\ket{\GHZ_N} = \frac{1}{\sqrt{2}}(\ket{0}^{\otimes N} + \ket{1}^{\otimes N})
\end{align}
is a desirable candidate, since it yields a quantum advantage in
metrology and conference key agreement \cite{giovannetti2011advances,Murta2020}. On the other hand, generating a two-dimensional cluster state allows us to implement schemes for measurement-based 
quantum computation \cite{briegel2009measurement}
or the establishment of a data bus \cite{freund2024flexible}.
Both families of states are part of the family of graph 
states \cite{hein2004multiparty}, which will therefore be the main 
interest of our work.  

None of the above-mentioned tasks can be successfully 
performed without the ability to generate these states 
with high fidelity. In particular, uncorrelated product 
states (that can arise in classical networks) do not provide 
any quantum advantage for these applications, making their 
use ineffective. Thus, the relevant question is: how well 
can graph states be generated in quantum networks? Do LOSR 
networks provide any advantage beyond classical networks?

To assess the suitability of the LOSR paradigm for large-scale quantum 
state preparation, we consider the fidelity as a popular and physically
relevant measure of the distinguishability of quantum states~\cite{nielsen2010quantum}.
First, we consider the simplest triangle scenario, see also Fig. \ref{fig:losr3} and the creation of tripartite
GHZ states. We present optimized protocols for GHZ generation, but also
derive rigorous upper bounds on the achievable fidelities. Subsequently, 
we show that the bounds derived in the triangular scenario can be lifted 
to fidelity constraints in arbitrarily large networks, thereby strictly 
limiting the amount of multipartite entanglement within such networks.
Interestingly, besides having an impact on quantum network design, this 
also may shed light on some discussions on notions of multipartite entanglement in the literature {(see e.g.~Ref.~\cite{navascues2020genuine} and Ref.~\cite{tavakoli2022bell} sections VI.~E--F)}.

\section{Results}
\subsection{Triangle networks}

We start with the simplest nontrivial quantum network of three 
parties, assuming the framework of local operations and shared 
randomness. In this scenario, three spatially separated participants, 
called Alice, Bob and Charlie ($A$, $B$ and $C$), are connected by 
bipartite source states shared between each pair. The parties apply 
local quantum operations on their subsystems, assisted by a shared 
classical random variable. More formally, states that can be prepared 
in a triangle LOSR network are of the form
\begin{align}
\label{eq:losrstate}
\rho=\sum_{\lambda}p(\lambda)\mathcal{E}_A^{\lambda}\otimes \mathcal{E}_B^{\lambda}\otimes \mathcal{E}_C^{\lambda}[\rho_{ab^\prime}\otimes\rho_{bc^\prime}\otimes \rho_{ca^\prime}],
\end{align}
where $\mathcal{E}_A^{\lambda}$,  $\mathcal{E}_B^{\lambda}$,  
$\mathcal{E}_C^{\lambda}$ are local (trace-preserving) operations, 
$\rho_{ab^\prime}$, $\rho_{bc^\prime}$, $\rho_{ca^\prime}$ are 
the source states and $p(\lambda)$ is a classical probability 
distribution (see Fig.~\ref{fig:losr3} and 
Refs.~\cite{navascues2020genuine, hansenne2022symmetries} for 
some technical details). This definition can be extended to an 
arbitrary number of parties by considering quantum networks where all participants are connected by bipartite source states. In the following, 
we call states that can be brought into the form of Eq.~\eqref{eq:losrstate} 
$\LOSR$ states. Note that we do not 
limit the dimension of the bipartite source states $\rho_{ij}$.

\begin{figure}[t]
    \centering
    \includegraphics[width=0.79\linewidth]{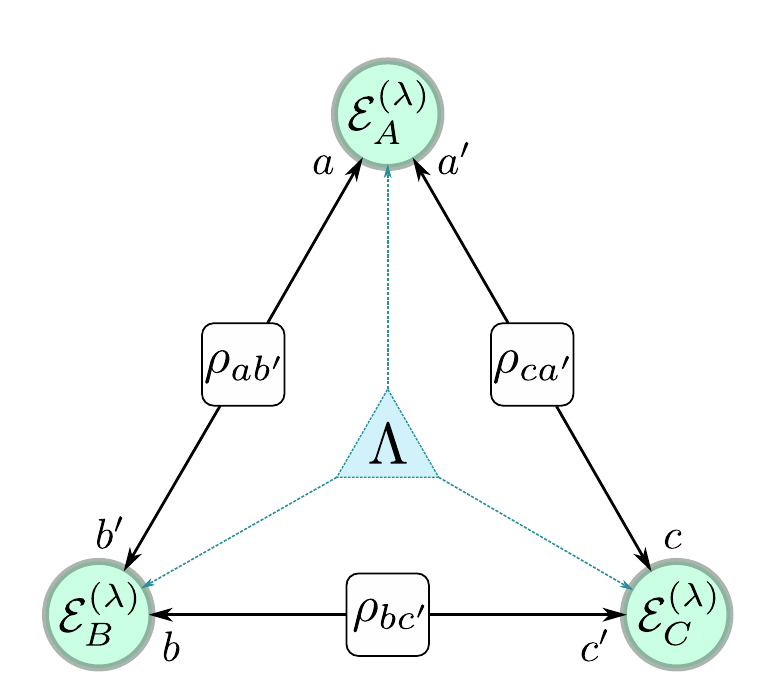}
    \caption{The LOSR (local operations and shared randomness) scenario 
    of three-partite networks. The three parties $A$, $B$ and $C$ are 
    receiving particles from bipartite sources $\rho_{ij}$ and are 
    allowed to perform local operations $\EE_A^{(\lambda)}, \EE_B^{(\lambda)}, \EE_C^{(\lambda)}$, based on a shared random variable $\lambda$, 
    to produce a three-partite output state.}
    \label{fig:losr3}
\end{figure}

To analyze the amount of multipartite entanglement one can generate 
we consider the maximal fidelity of an LOSR network state $\rho$ with 
strongly entangled multipartite states, such as the multipartite qudit 
GHZ state $\ket{\GHZ_{N,d}}$, given by
\begin{align}
F_{\LOSR}(\ket{\GHZ_{N,d}}) = 
\max_{\rho \in \LOSR} 
\braket{\GHZ_{N,d}|\rho|\GHZ_{N,d}},
\end{align}
where 
$
\ket{\GHZ_{N,d}} =  
(\sum_{i=0}^{d-1} \ket{i}^{\otimes N})/\sqrt{d},
$
is the $d$-level GHZ state of $N$ parties (if $d=2$, we omit the dimension and write $\ket{\GHZ_N}$ instead). 
Of course, this is a specific target state,  
but we will later show that much larger classes
of target states can be reduced to the GHZ case.
In the case of the GHZ state, a fidelity value of $F={1}/{d}$ 
can directly be achieved by preparing the product state 
$\ket{0}^{\otimes N}$. In contrast, any value above 
$1/d$ already indicates the presence of genuine 
multipartite entanglement~\cite{bourennane2004experimental}, 
and a value of $F_{\LOSR}(\ket{\GHZ_{N,d}}) = 1$  would indicate that the GHZ 
state can be prepared perfectly. 

Let us start our discussion with the three-qubit case.
For this, using numerical semidefinite programming techniques together with so-called inflation arguments, Ref.~\cite{navascues2020genuine} showed that 
$F_{\LOSR}(\ket{\GHZ_{3}}) \leq (1+\sqrt{3})/4 < 0.684$. 
This raises the question of whether LOSR quantum 
networks can reach this bound. In fact, the authors 
of Ref.~\cite{navascues2020genuine} also  report that 
they numerically obtained a LOSR state with a 
fidelity of approximately $0.517$. On the one hand, this is 
remarkable, as it shows that genuine multipartite 
entanglement can be prepared in the LOSR scenario, 
on the other hand this begs for improvement and 
the construction of explicit network protocols.

So, we designed constructive protocols for the preparation
of three-qubit GHZ states in LOSR networks, see the Appendix~\ref{app:constructions} for details. For the 
case of two-qubit source states $\rho_{ij}$ we find the 
optimal protocol, giving a GHZ fidelity of $F= [5+4\cos(2\pi/7)]/[12+4\cos(2\pi/7)] \approx 0.517$, 
confirming the reported numerical value. Surprisingly, however, 
the fidelity can be systematically increased if higher-dimensional 
resource states are used. Studying our optimized protocols for up 
to ten-dimensional source states, the achievable fidelities 
converge to $F \approx 0.548$. The fact that the reached fidelities increase monotonically with the source dimension suggests that the fidelity can act as a dimensionality witness for the sources.

On the other hand, we developed a novel approach for
obtaining upper bounds on the achievable fidelities.
This is based on an interpretation of the fidelity as
an optimization over generalized measurements  \cite{bengtsson2017geometry}, 
combined with the Finner inequality \cite{renou1} (see also Appendices~\ref{app:analyticalbound} and \ref{app:fidbound_ghzn}) 
which for the qubit case leads to $F < 0.618.$ This allows us to summarize our first main result:

\begin{result}\label{res:GHZ32bound}
In LOSR networks of three nodes the maximal achievable fidelity 
with the three-qubit GHZ state is bounded by
\begin{align}
\label{eq:boundtripartite2}
        0.548 < F_{\LOSR}(\ket{\GHZ_3}) < 0.618.
\end{align}\vspace{-1.5em}
\end{result}

Before discussing more general classes of quantum states and larger 
networks, let us formulate what our methods can teach about the 
generation of higher-dimensional GHZ states in triangle networks. 
In fact, as outlined in the Appendices, the methods can, with some
modifications, also be applied to these cases. In our optimized 
network protocols for generating the states, however, it is relevant
whether the dimension is odd or even. All our results can be summarized
as follows: 

\begin{result}\label{res:GHZ3dbound}
In LOSR networks of three nodes the maximal achievable fidelity with the three-qudit GHZ state for $d \geq 3$ is bounded by
\begin{align}
    \label{eq:boundtripartited}
        \frac{1}{d}\left(\frac{3}{2} - \frac{\lfloor{d}/{2}\rfloor}{d(d-1)} \right) \leq F_{\LOSR}(\ket{\GHZ_{3,d}}) \leq \frac{1}{\sqrt{d}},
    \end{align}
where $\lfloor \cdots \rfloor$ indicates the floor function.
\end{result}
Note that Result~\ref{res:GHZ3dbound} also applies to $d=2$, but we have stronger bounds in this case with Result~\ref{res:GHZ32bound}.
Since the fidelities of the constructed states scale like $3/(2d)$ and 
exceed the value of $1/d$, our constructive protocols show 
that the preparation of genuine multipartite entanglement is possible 
for all output dimensions. In contrast, the difference between the 
upper bound and the threshold $1/d$ is not large and vanishes 
in the asymptotic limit.

\subsection{Networks of arbitrary size}

In the following, we introduce a method to treat networks of arbitrary size. The basic idea is the following (see also Fig.~\ref{fig:ghz_extraction_cluster} (a)): consider a network 
of $N$ nodes with bipartite sources, which is able to prepare a quantum state $\rho_N$. We can group the $N$ nodes into three groups, $A$, $B$, and $C$, take local transformations $\EE_A$, $\EE_B$, and $\EE_C$ acting locally on these groups and then consider the state
\begin{equation}
\EE_A \otimes \EE_B \otimes \EE_C (\rho_N) = 
\rho_{\triangle} .
\end{equation}
Clearly, since $\rho_N$ can be generated in the network, 
the state $\rho_{\triangle}$ can be prepared in the 
triangle network formed by the three groups $A$, $B$, and $C$. Conversely, if $\rho_{\triangle}$
is not feasible, e.g., due to our results from the previous
section, any state $\rho_M$ which can be reduced to $\rho_{\triangle}$ via some maps $\EE_A$, $\EE_B$, and $\EE_C$, cannot be prepared in the full network of $N$ nodes with bipartite 
sources. 

To make this idea a powerful tool, we use three 
specifications. First, we focus on pure states and
consider unitary transformations as maps. More precisely, 
for a given $N$ particle state $\ket{\psi_N}$ we aim to find three groups 
and unitaries such that
\begin{align} 
\label{eq:GHZextraction}
U_A \otimes U_B \otimes U_C 
\ket{\psi_N} = 
\ket{\phi_\triangle}_{abc} \otimes \ket{\tau},
\end{align}
where $a \subseteq A$, $b \subseteq B$, and 
$c \subseteq C$ are subsets of the groups and 
$\ket{\tau}$ is a residual state on the remaining
nodes. By tracing out these nodes, the state 
$\ket{\phi_\triangle}_{abc}$ can be prepared 
in the triangle network, if $\ket{\psi_N}$ can be
prepared in the large network. Importantly, this 
ansatz gives a direct relation between reachable
fidelities in networks. If a state $\rho_N$
can be prepared in the large network, having a fidelity
$F$ with the state $\ket{\psi_N}$, then one can prepare
a state $\rho_\triangle$ in the triangle scenario, 
having at least the same fidelity $F$ with $\ket{\phi_\triangle}_{abc}$.

Second, we want to use our fidelity bounds on the
GHZ state in triangle networks to derive limitations for
general networks. Third, we can employ a specific 
property of GHZ states: multiple copies of a three-qubit GHZ 
state are unitarily equivalent to a single GHZ state in 
higher dimensions~\cite{kraft2018characterizing}
\begin{align}
\ket{\GHZ_3}^{\otimes k} \cong \ket{\GHZ_{3,2^k}}.
\end{align}
Consequently, if several copies of GHZ states can be prepared 
in the tripartitioned network, one can apply
the exponentially decaying fidelity bounds for high-dimensional GHZ states from Result 2. 

All that remains to be done to understand the general limitations of large quantum 
networks is to identify families of quantum states, which can be transformed into
three-partite GHZ states using the local transformation on three groups of nodes
as outlined above. Here, the family of quantum graph states comes into play. 
Graph states form a family of multi-qubit quantum states, comprising multipartite
GHZ states, as well as cluster states \cite{schlingemann2001quantum, raussendorf2003measurement, hein2004multiparty}. For defining graph states one
considers a graph, consisting of nodes which are linked with some edges. Then, 
one associates each node with a qubit, while edges correspond to entangling 
two-qubit gates carried out between two nodes, to arrive at a pure quantum state corresponding
to the graph (the formal definitions are given in Appendix~\ref{app:graphstates}).
Graph states are particularly relevant in the context of quantum networks: protocols such 
as measurement-based quantum computation \cite{briegel2009measurement},
blind quantum computation \cite{broadbent2009universal}, or entanglement routing \cite{freund2024flexible}
require the distribution of graph states to 
distant parties. Naturally, graph states have previously been studied in the context 
of LOSR networks \cite{hansenne2022symmetries, makuta2023no, wang2024quantum}, but the only known general fidelity bound 
states that arbitrary graph states cannot be 
prepared in networks using bipartite sources with fidelities exceeding $0.9$ \cite{wang2024quantum}.

Remarkably, for graph states the number of extractable GHZ states for a fixed tripartition has been studied and is known as an entanglement quantifier 
\cite{bravyi2006ghz, linden2002almost}. 
In these works, however, it remained unclear whether every graph state admits a tripartition that yields at least one GHZ state,
but in Appendix~\ref{app:graphstates} we demonstrate that this is indeed the case.  This directly leads to:
\begin{figure}[t]
    \centering
    \includegraphics[width=0.8\linewidth]{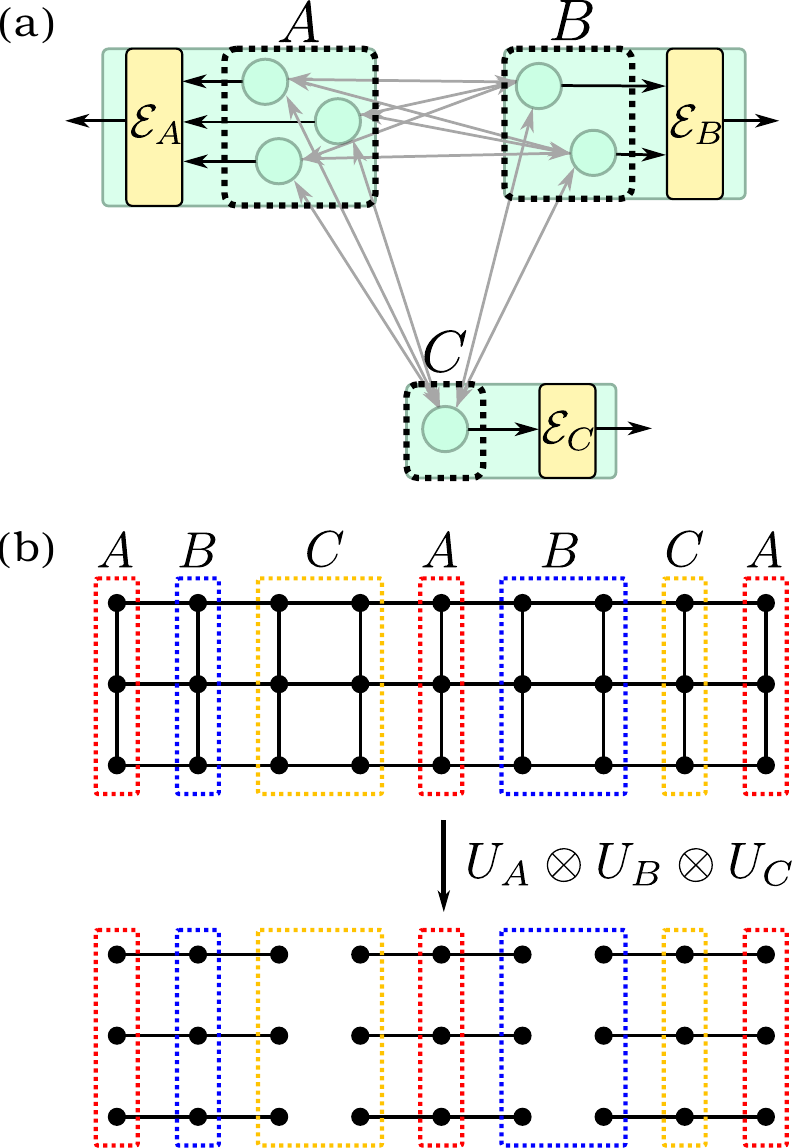}
    \caption{%
    (a) Visualization of the reduction of an $N$-partite network (here, $N=6$) to a triangle network by combining several nodes into one. If a state $\rho_N$ is generated by the $N$-partite network, then $\rho_\triangle = \EE_A\otimes \EE_B \otimes \EE_C(\rho_N)$ must be generable in a tripartite network.
    (b) Example of GHZ extraction from a cluster state. 
    The initial $27$-qubit cluster state $\ket{\text{Cl}_{3,9}}$ can be transformed into $\ket{\GHZ_3}^{\otimes 9}$ using a tripartition of the parties into three sets $A$ (indicated red), $B$ (blue) and $C$ (orange). Within the sets, the unitaries $U_A$, $U_B$ and $U_C$ are chosen such that they remove links, leaving decoupled graph states of three connected parties which are known to be equivalent to $\ket{\GHZ_3}$.
    }
    \label{fig:ghz_extraction_cluster}
\end{figure}
\setcounter{thmc}{1}
\begin{result} \label{res:graphstatebound}
   No graph state corresponding to a graph with at least three connected nodes  can be prepared in an LOSR 
    network with bipartite sources with a fidelity larger than $F_{\LOSR}(\ket{\GHZ_3}) < 0.618$.
\end{result}
To strengthen this result we note that many graph states 
admit partitions that can produce more than one GHZ state.
This includes the family of two-dimensional cluster states 
$\ket{\text{Cl}_{m,n}}$, which correspond to a $N=m \times n$
grid as a graph, see Fig.~\ref{fig:ghz_extraction_cluster} (b).
These are known to be resource states for measurement-based 
quantum computation \cite{briegel2009measurement}.
For such states, it is always possible to find a tripartition yielding the maximal number of $k=\lfloor mn/3\rfloor$ GHZ 
states. 
Moreover, a lower bound on the achievable fidelity in LOSR networks can be obtained from the knowledge of the maximal fidelity of 
fully separable product states with the cluster state \cite{markham2007entanglement}. Together, we obtain
\begin{result}\label{res:clusterstatebound}
    The maximal fidelity of the $m\cdot n$-partite cluster state $\ket{\text{Cl}_{m,n}}$ in LOSR networks of $N$ parties and bipartite sources is bounded by
    \begin{align}
        \frac{1}{2^{\lfloor \frac{mn}{2}\rfloor}} \leq F_{\LOSR}(\ket{\text{Cl}_{m,n}}) \leq \frac{1}{2^{\frac12\lfloor \frac{mn}{3}\rfloor}}.
    \end{align}
    \vspace{-1.5em}
\end{result}
The fact that the maximal achievable fidelity is suppressed exponentially 
with the number of qubits suggests that LOSR networks do not provide any advantage compared to classical networks for tasks based on cluster states.

\subsection{Multipartite GHZ states}

Our previous results prompt two immediate questions. 
First, we 
focused so far on the ability to extract three-partite states, but one 
could ask for the extraction of more-partite states, e.g.~$\ket{\GHZ_4}$ instead, 
if better fidelity bounds are known for these states. Although it 
is often the case, it is, however, not always possible to find tetrapartitions 
that extract $\ket{\GHZ_4}$ from connected graph states \cite{englbrecht2022transformations}.
A second question asks for direct fidelity bounds for  
$\ket{\GHZ_{N,d}}$ for $N>3$ without resorting to any GHZ
extraction procedure. For $N\geq4$, stronger bounds than the ones
for the triangle scenario can be established. They imply
that in these LOSR networks the fidelity of GHZ states 
cannot exceed the fidelity that can be obtained without
distributing bipartite entanglement at all. 

\begin{result}\label{res:multipartitebound}
    In quantum networks of $N \geq 4$ nodes with bipartite sources, the maximal fidelity with $\ket{\GHZ_{N,d}}$ is bounded by $1/d$. The same fidelity can be reached by a fully separable state in the classical network.
\end{result}

The proof of the statement can be found in Appendix~\ref{app:fidbound_ghzn}. 
Note that this fidelity bound allows to construct novel extraction
scenarios, aiming at the extraction of four qubit GHZ states. In this
sense, the results of Ref.~\cite{englbrecht2022transformations} may be used to 
provide more bounds for general graph states. 

\section{Discussion}

Our results demonstrate a fundamental limit in the state 
preparation capabilities of LOSR quantum networks with bipartite sources. Specifically, we have shown that these networks are inherently unable to generate entangled multipartite states with fidelities that significantly exceed those of fully separable 
states.
To do so, we also provided insights into GHZ state extraction  from graphs states that are along the lines of Refs.~\cite{englbrecht2022transformations, de2024extracting}. We believe these results are of independent interest. 

Our results enforce a reassessment of the potential of the LOSR 
paradigm for distributed quantum information processing and 
raise the need to explore alternative network architectures 
that rely on classical communication. Since the usage of 
classical communication typically introduces time delays, adding quantum memories may be required. To deal with 
this, a crucial next step is to determine the minimal amount 
of classical communication required for preparing highly 
entangled states in networks. This would allow us to derive 
concrete estimates on memory usage times and waiting times \cite{collins2007multiplexed, shchukin2019waiting, weinbrenner2024aging}, setting requirements for practical 
and scalable experimental implementations. Moreover, this may pave
the way towards more efficient schemes for entanglement distillation \cite{dur2007entanglement, rozpkedek2018optimizing}. In fact, understanding the minimal resources required for reaching the full power of quantum networks
may guide the ultimate development of the quantum internet.

{\it Note added.---} While finishing this work we learned that related but complementary results have been obtained by Xiang Zhou et al.

\section*{Acknowledgments}
\acknowledgments
We thank  Ghislaine Coulter-de Wit,  Mariami Gachechiladze, Tristan Kraft, Julia Kunzelmann, Nikolai Miklin, Anton Trushechkin and Lina Vandré for discussions. 

This work has been supported by the Deutsche Forschungsgemeinschaft 
(DFG, German Research Foundation, project numbers 447948357 and 440958198),
the Sino-German Center for Research Promotion (Project M-0294), the German Ministry of Education and Research (Projects QuKuK, BMBF Grant No.~16KIS1618K and QSolid, BMBF Grant No.~ 13N16163).
K.H. and L.T.W. acknowledge support by the House of Young Talents of the University of Siegen.

\bibliographystyle{apsrev4-2}
\bibliography{sources}

\clearpage
\onecolumngrid

\appendix

\pagestyle{fancy}
\fancyhf{}
\fancyhead[R]{\thepage}
\fancyhead[L]{\leftmark}

\section{Protocols to achieve high GHZ fidelity} \label{app:constructions}

\begin{figure}[b]
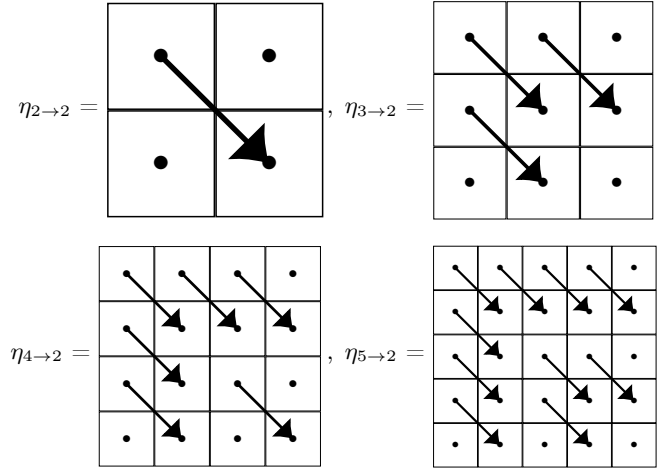

    \centering
    \begin{subfloat}
    
        $\eta_{2\rightarrow2}$ = \raisebox{0.75cm}{\resizebox{3cm}{!}{\etatwotwo}} $\!\!\!\!,\;\eta_{3\rightarrow2}$ = \raisebox{1cm}{\resizebox{3cm}{!}{\etathreetwo}}
    \end{subfloat}
    \\
    \phantom{Igel}
    \\
    \begin{subfloat}
    
        $\eta_{4\rightarrow2}$ = \raisebox{1.125cm}{\resizebox{3cm}{!}{\etafourtwo}} $\!\!,\;\eta_{5\rightarrow2}$ = \raisebox{1.2cm}{\resizebox{3cm}{!}{\etafivetwo}}
    \end{subfloat}
    
    \caption{Visual grid representations of the Choi matrices $\eta_{\din\rightarrow2}$ defined in Eqs.~\eqref{eq:etadim2simple} and \eqref{eq:etadim2} that yield the fidelities with $\ket{\GHZ_3}$ reported in Table~\ref{tab:fidelities_din2}. Each arrow connecting cell $(i,j)$ to cell $(k,l)$ corresponds to a term $\tau_{(i,j),(k,l)}$ as defined in Eq.~\eqref{eq:taudef}, whereas a single dot in an unconnected cell $(m,n)$ corresponds to a term $\tau_{(m,n)}$. Note that every cell has to be covered by exactly one arrow or dot in order to obtain a trace preserving channel. }
    \label{fig:optimaletas}
\end{figure}

In this section, we present our construction for specific channels $\EE_A$, $\EE_B$, $\EE_C$ and distributed states $\rho_{ab'}$, $\rho_{bc'}$, $\rho_{ca'}$ that yield network states $\rho$ according to Eq.~\eqref{eq:losrstate} with large $\ket{\GHZ_3}$-fidelities. Throughout this section, we denote the local dimension of the distributed states by $\din$ and the dimension of the output of the channels by $\dout$.
As the target state is pure, the fidelity is a linear function and the maximal fidelity can be reached by a network state that does not make use of shared randomness. Therefore, we only have to find single channels and source states.

Using the Choi-Jamiołkowski isomorphism, the action of a channel $\EE_A$ from input systems $aa'$ to an output $\alpha$ can be written as
\begin{align}
\EE_A(\rho_{aa'}) = \trace_{aa'}[\eta_A (\rho_{aa'}^T \otimes \one_\alpha)],
\end{align}
where
\begin{align}
    \eta_A = \sum_{i,j,k,l=0}^{\din-1} \ketbraa{ij}{kl}_{aa'} \otimes \EE_A(\ketbraa{ij}{kl})_\alpha
\end{align}
where $\eta_A$ is the so called Choi matrix. The fact that $\mathcal{E}_A$ is a completely positive, trace preserving map translates into constraints on the Choi matrix $\eta_A$, namely 
\begin{align}
    \eta_A &\geq 0, \label{eq:choiconst1}\\ \trace_\alpha(\eta_A) &= \one_{aa'} \label{eq:choiconst2},
\end{align}
thus, $\trace(\eta_A) = \din^2$.

Using the Choi matrices, we can write the output state $\rho_{\alpha\beta\gamma}$ of the network as
\begin{align} \label{eq:choiitn}
    \rho_{\alpha\beta\gamma} &= \mathcal{E}_A \otimes \mathcal{E}_B \otimes \mathcal{E}_C(\rho_{ab'} \otimes \rho_{bc'} \otimes \rho_{ca'})\nonumber  \\
    &= \trace_{\substack{aa'\\bb'\\cc'}}[(\eta_A\otimes \eta_B \otimes \eta_C)(\rho_{ab'}^T \otimes \rho_{bc'}^T \otimes \rho_{ca'}^T \otimes \one_{\alpha\beta\gamma}) ].
\end{align}
Note that the order of systems in $\eta_A\otimes \eta_B \otimes \eta_C$ is different from the order of systems in $\rho_{ab'}^T \otimes \rho_{bc'}^T \otimes \rho_{ca'}^T \otimes \one_{\alpha\beta\gamma}$.

\subsection{Grid channels}
\label{sec:gridtype}

Let us introduce a framework to represent the constructed Choi matrices in a graphical way using a grid. We call the resulting maps of this construction \emph{grid channels}, as their definition shares some similarities with the notion of grid states \cite{lockhart2018entanglement, ghimire2023quantum, krebs2024high}. 

In order to form a grid channel, we represent its Choi matrix by a $\din \times \din$-dimensional grid as visualized in Fig.~\ref{fig:optimaletas}. We then connect the cells of the grid by arrows of different length. Each arrow of length $k$ connecting cells $(i_0,j_0), \ldots, (i_{k-1},j_{k-1})$ translates into a contribution of 
\begin{align} \label{eq:taudef}
    \tau_{(i_0,j_0), \ldots, (i_{k-1},j_{k-1})} \coloneq  \ketbra{T_{(i_0,j_0), \ldots, (i_{k-1},j_{k-1})}}
\end{align} to the Choi matrix, where
\begin{align} \label{eq:Tdef}
    \ket{T_{(i_0,j_0), \ldots, (i_{k-1},j_{k-1})}} \coloneq  \sum_{l=0}^{k-1} \ket{i_l, j_l, l}.
\end{align}
Here, the first two systems correspond to the input, whereas the third system corresponds to the output of the channel. This indicates that the maximal length of an arrow is given by the output dimension $\dout$. We do allow for arrows of length one, which are represented by a single dot.
This construction inherently yields a positive semidefinite Choi matrix, fulfilling constraint~\eqref{eq:choiconst1}. The second constraint~\eqref{eq:choiconst2}, i.e.~trace preservation, is ensured by covering each cell on the grid by exactly one arrow. 

As an example, consider the grid in Fig.~\ref{fig:optimaletas} for $\din = \dout = 2$. It represents the Choi matrix
\begin{align} \label{eq:etadim2simple}
    \eta_{2\rightarrow 2} = (\ket{000} + \ket{111})(\bra{000}+\bra{111}) 
    + \ketbra{010} + \ketbra{100}.
\end{align}

\subsection{Constructions for $\dout = 2$}
\label{sec:constd2}

Let us provide a construction of channels for
arbitrary input dimensions and fixed output dimension $\dout=2$. In Fig.~\ref{fig:optimaletas}, we show the grid representation of these channels for  $\din = 2,3,4$ and $5$. However, we can generalize the construction to arbitrary input dimension via
\begin{align}\label{eq:etadim2}
    \eta_{\din\rightarrow2} = \sum_{r=0}^{\lfloor\frac{\din}{2} \rfloor -1} \big[\tau_{(2r,2r),(2r+1,2r+1)} +  
    \sum_{k=1}^{\din-2-2r}\!\!\!\!\!\! (\tau_{(2r+k,2r),(2r+k+1,2r+1)} + \tau_{(2r,2r+k),(2r+1,2r+k+1)})\big]
\end{align}
In the next Appendix~\ref{app:fiddout2}, we analytically calculate the GHZ-fidelity that is achieved by choosing these grid-type channels as a function of the Schmidt coefficients of the input states. Here, we summarize the results of this calculation.
To that end, we use Eq.~\eqref{eq:choiitn} to express the fidelity for a specific choice of $\eta_A=\eta_B=\eta_C \equiv \eta$ as
\begin{align}
    F_{\eta} \coloneq   \trace[\eta^{\otimes 3}(\rho_{ab'}^T \otimes \rho_{bc'}^T \otimes \rho_{ca'}^T \otimes \ketbra{\GHZ_3}_{\alpha\beta\gamma}) ].
\end{align}
It turns out that $F_{\eta}$ can be calculated by counting cycles of length $3$ in a graph induced by the choice of grid for the Choi matrix. In fact, a guiding principle in choosing specifically the grids presented in Fig.~\ref{fig:optimaletas} is to maximize this number of cycles, as each cycle contributes to the total fidelity.

We present a recursive expression for all input dimensions and input states $\ket{\psi}_{ab'} = \ket{\psi}_{bc'} = \ket{\psi}_{ca'} = \sum_{i=0}^{\din-1} \lambda_i \ket{ii}$ in Theorem~\ref{thm:fidelitydin2} in Appendix~\ref{app:fiddout2dingr2}. For $\din=2$, the fidelity reads
\begin{align}\label{eq:fidelity_schmidt_22}
    F_{\eta_{2\rightarrow2}} = \frac12\left(\lambda_0^6+\lambda_1^6+ 3\lambda_0^4\lambda_1^2 + 2\lambda_0^3 \lambda _1^3\right).
\end{align}
Its maximal value is achieved for $\lambda_0^2 = 19/39 + 7/39 (1/z^{1/3} + z^{1/3}) \approx 0.8346811596$ with $z=71/98 + 39\sqrt{3}/98\cdot i$ and evaluates to $[5+4\cos(2\pi/7)]/[12+4\cos(2\pi/7)] \approx 0.5170401751$. 
This matches the reported value in Ref.~\cite{navascues2020genuine}. Furthermore, we checked its optimality by using a seesaw-based optimization technique using semidefinite programming, as simultaneous optimization over states and channels is not convex and thus inefficient.
We thus fix five of the six matrices (three input states and three Choi matrices) and optimize only one at a time. This allows us to obtain quickly converging achievable fidelities for small input and output dimensions and confirms that our construction is likely optimal.

It is quite remarkable that in order to achieve the above-mentioned  value one has to distribute weakly entangled states. In fact, fixing the input states to maximally entangled Bell pairs, we could not find any channel achieving a $\ket{\GHZ_3}$-fidelity larger than $1/2$ (recall that only fidelities larger than $1/2$ are non-trivial, as the product state $\ket{000}$ yields already a fidelity of $1/2$).
This vaguely resembles effects observed in preparation schemes for four-partite cluster states, where less entangled bipartite resource states lead to higher success probabilities \cite{zhang2016experimental}.

Also, using larger input dimensions yields an advantage: In Table~\ref{tab:fidelities_din2}, we list the achievable $\ket{\GHZ_3}$-fidelity for different input dimensions. It is apparent that each additional input dimension yields an advantage. However, the fidelity quickly converges to a value slightly larger than $0.548047$, thereby providing a proof of the lower bound stated in Result~\ref{res:GHZ32bound}.

We give the explicit form of the resulting three-qubit state achieving this fidelity in Eq.~\eqref{eq:rhod2} in Appendix~\ref{app:fiddout2dingr2}.
We conjecture that this is the highest $\ket{\GHZ_3}$-fidelity achievable in triangle LOSR networks with two-dimensional output. Furthermore, the reported values in Table~\ref{tab:fidelities_din2} are matched by our numerical seesaw method up to input dimension $\din = 5$ (we could not check larger dimensions). 

\subsection{General fidelity expression for $\dout = 2$}\label{app:fiddout2}

Let us complete the above argument by calculating an analytical expression for the obtained $\ket{\GHZ_3}$-fidelity with $\ket{\GHZ_3}$ in the triangle network with $\dout = 2$ using the channels presented in Eq.~\eqref{eq:etadim2}. This then provides a proof for the lower bound reported in Result~\ref{res:GHZ32bound} from the main text.

\subsubsection{Fidelity for $\din = 2$}
\label{app:fiddout2din2}

It is illustrative to start with the simplest case of $\din = 2$ where
\begin{align} \label{eq:eta22_appendix}
    \eta_{2\rightarrow 2} = (\ket{000} + \ket{111})(\bra{000}+\bra{111}) + \ketbra{010} + \ketbra{100}.
\end{align}
Note that in this expression, the first two systems correspond to the inputs $a$ and $a'$, whereas the last system corresponds to the output $\alpha$. 
We start by rearranging the expression by pulling the output system to the front and stacking the input system on top of each other, i.e., $\ketbrau{i}jkl_a^{a'} \coloneq  \ketbraa{i}{k}_{a'} \otimes \ketbraa{j}{l}_{a}$. Thus, for system $A$ we get
\begin{align}
    \eta_{2\rightarrow 2} = \ketbraa{0}{0}_\alpha \otimes \left( \ketbrau0000_{a}^{a'} + \ketbrau0101_{a}^{a'} + \ketbrau1010_{a}^{a'}\right) \;+\; \ketbraa{0}{1}_\alpha \otimes \ketbrau0011_{a}^{a'}  \;+\;  \ketbraa{1}{0}_\alpha \otimes \ketbrau1100_{a}^{a'}  \;+\;  \ketbraa{1}{1}_\alpha \otimes \ketbrau1111_{a}^{a'}.
\end{align}
Likewise expanding systems $B$ and $C$ yields
\begin{alignat}{5}\label{eq:eta22expanded}
     \eta_{2\rightarrow 2}^{\otimes 3} = &\big[ \ketbraa{0}{0}_\alpha \otimes ( \ketbrau0000_{a}^{a'} + \ketbrau0101_{a}^{a'} + \ketbrau1010_{a}^{a'}) &\;+\; \ketbraa{0}{1}_\alpha \otimes \ketbrau0011_{a}^{a'}  &\;+\;  \ketbraa{1}{0}_\alpha \otimes \ketbrau1100_{a}^{a'}  &\;+\;  \ketbraa{1}{1}_\alpha \otimes \ketbrau1111_{a}^{a'} \big] &\nonumber \\
     \otimes & \big[\ketbraa{0}{0}_\beta \otimes ( \ketbrau0000_{b}^{b'} + \ketbrau0101_{b}^{b'} + \ketbrau1010_{b}^{b'}) &\;+\; \ketbraa{0}{1}_\beta \otimes \ketbrau0011_{b}^{b'}  &\;+\;  \ketbraa{1}{0}_\beta \otimes \ketbrau1100_{b}^{b'}  &\;+\;  \ketbraa{1}{1}_\beta \otimes \ketbrau1111_{b}^{b'}\big] &\nonumber \\
     \otimes & \big[\ketbraa{0}{0}_\gamma \otimes ( \ketbrau0000_{c}^{c'} + \ketbrau0101_{c}^{c'} + \ketbrau1010_{c}^{c'}) &\;+\; \ketbraa{0}{1}_\gamma \otimes \ketbrau0011_{c}^{c'}  &\;+\;  \ketbraa{1}{0}_\gamma \otimes \ketbrau1100_{c}^{c'}  &\;+\;  \ketbraa{1}{1}_\gamma \otimes \ketbrau1111_{c}^{c'}\big] &.
\end{alignat}
The input states are chosen to be given in the Schmidt diagonal form
\begin{align}
    \ketbra{\psi}^{\otimes 3} = \ketbra{\psi}^a_{b'} \otimes \ketbra{\psi}^b_{c'} \otimes \ketbra{\psi}^c_{a'}    
\end{align}
with
$\ket{\psi}^a_{b'} = \lambda_0 \ketu{0}{0}_{b'}^a + \lambda_1 \ketu{1}{1}_{b'}^a$, where the Schmidt coefficients $\lambda_i$ are real valued, and we will be taking the expectation value of the output with the GHZ state $\ket{\GHZ_3} = \frac1{\sqrt2}(\ket{000}_{\alpha\beta\gamma} + \ket{111}_{\alpha\beta\gamma})$ to form the fidelity
\begin{align}
    F_{\eta_{2\rightarrow2}} = \bra{\GHZ_3} \eta_{2\rightarrow 2}^{\otimes 3} \ketbra{\psi}^{\otimes 3} \ket{\GHZ_3} .
\end{align}
Note that we can omit the transposition of the input states as they are chosen to be real-valued. Due to the symmetries of the GHZ state, when considering the expansion of the expression in Eq.~\eqref{eq:eta22expanded}, only terms survive which have matching entries in the $\alpha$, $\beta$ and $\gamma$ components. In other words, we are only allowed to expand along the four columns in that expression.

\begin{figure}[t]
    \centering
    \resizebox{4cm}{!}{
    \begin{tikzpicture}[->]
        \node (0000) at (0,0){\huge$\ketbrau0000$};
        \node (0101) at (3,0) {\huge$\ketbrau0101$};
        \node (1010) at (0,-2.5){\huge$\ketbrau1010$};
    
        \path (0000) edge [loop above] node {$\lambda_0^2$} (0000);
        \path (0000) edge node[above] {$\lambda_0^2$} (0101);
        \path (0101) edge node[above,left,yshift=5pt] {$\lambda_1^2$} (1010);
        \path (1010.east) edge node[below,right,yshift=-5pt] {$\lambda_0^2$} (0101.south);
        \path (1010) edge node[left] {$\lambda_0^2$} (0000);
    \end{tikzpicture}
    }\quad\quad
    \raisebox{1.5cm}{
    \begin{tabular}{c|c|c}
        Cycle & Multiplicity & Contribution \\
        \hline 
        $\ketu00\!\!\tikzmark{a}\!\brau00 \rightarrow \ketbrau0000 \rightarrow \ketu00\!\tikzmark{b}\!\!\brau00
        \tikz[overlay,remember picture]
        {\draw[->,square arrow] (b.south) to (a.south);}$ & $1$ & $\lambda_0^6$ \vphantom{$\substack{1\\1\\1\\1\\1\\1}$} \\
        \hline
        $\ketu00\!\!\tikzmark{a}\!\brau00 \rightarrow \ketbrau0101 \rightarrow \ketu10\!\tikzmark{b}\!\!\brau10
        \tikz[overlay,remember picture]
        {\draw[->,square arrow] (b.south) to (a.south);}$ & $3$ & $\lambda_0^4\lambda_1^2$\vphantom{$\substack{1\\1\\1\\1\\1\\1}$}
    \end{tabular}
    }

    \caption{Left: The directed graph used to compute the contribution in Eq.~\eqref{eq:cont000_b} to the $\ket{\GHZ_3}$-fidelity yielded by the construction using the Choi matrices from Eq.~\eqref{eq:eta22_appendix}. Right: The cycles of size three of the graph on the left. The contribution to the $\ket{\GHZ_3}$-fidelity is determined by those cycles, where each of them contributes the product of its edge labels. } 
    \label{fig:graph_22}
\end{figure}
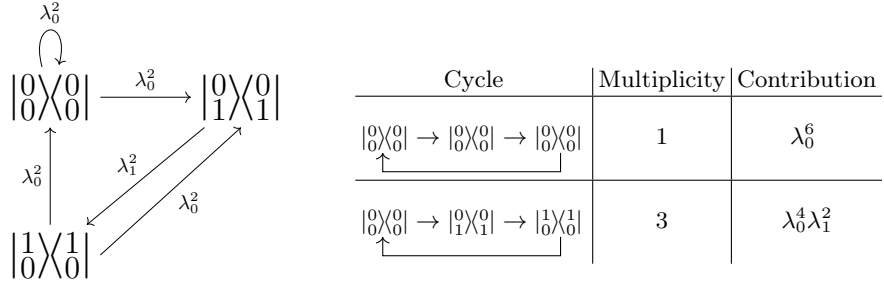

We will consider the four contributions to the fidelity individually, starting with 

\begin{align}\label{eq:cont000}
    \bra{000}_{\alpha\beta\gamma} \eta_{2\rightarrow 2}^{\otimes 3} \ketbra{\psi}^{\otimes 3} \ket{000}_{\alpha\beta\gamma} = 
    \begin{array}{c}
        \\
        \\
        \bra{\psi}^a_{b'} \\
        \otimes \\
        \bra{\psi}^b_{c'} \\
        \otimes \\
        \bra{\psi}^c_{a'}
    \end{array}
    \begin{pmatrix}
        (\ketbrau0000_{a}^{a'} + \ketbrau0101_{a}^{a'} + \ketbrau1010_{a}^{a'}) \\
        \otimes \\
        (\ketbrau0000_{b}^{b'} + \ketbrau0101_{b}^{b'} + \ketbrau1010_{b}^{b'}) \\
        \otimes \\
        (\ketbrau0000_{c}^{c'} + \ketbrau0101_{c}^{c'} + \ketbrau1010_{c}^{c'})
    \end{pmatrix}
    \begin{array}{c}
        \\
        \\
        \ket{\psi}^a_{b'} \\
        \otimes \\
        \ket{\psi}^b_{c'} \\
        \otimes \\
        \ket{\psi}^c_{a'}
    \end{array}.
\end{align}
We have offset the contraction with the input states to indicate that they connect between the rows of the expansion of the Choi matrices. Now, due to the symmetries of the Schmidt diagonal $\ket{\psi} = \lambda_0 \ketu00 + \lambda_1 \ketu11$, when expanding Eq.~\eqref{eq:cont000}, only terms survive where we select terms from the Choi matrix where the entries of the $a$ input and the $b'$ input match, and likewise for $b$ and $c'$, and $c$ and $a'$. This can be represented as a directed graph indicating allowed contractions as displayed in Fig.~\ref{fig:graph_22}, where we label the edges with the corresponding $\lambda_i$ factor. Now, the complete contribution of $\bra{000}_{\alpha\beta\gamma} \eta_{2\rightarrow 2}^{\otimes 3} \ketbra{\psi}^{\otimes 3} \ket{000}_{\alpha\beta\gamma}$ corresponds to the sum of products of edge labels of all cycles of length three in that graph. For the graph at hand, the three-cycles and the corresponding factors are indicated in the table on the right-hand side of Fig.~\ref{fig:graph_22}. Note that the second cycle comes with a factor of three as it can be started at any of its nodes. Counting these contributions, we get
\begin{align}\label{eq:cont000_b}
    \bra{000}_{\alpha\beta\gamma} \eta_{2\rightarrow 2}^{\otimes 3} \ketbra{\psi}^{\otimes 3} \ket{000}_{\alpha\beta\gamma} = \lambda_0^6 + 3\lambda_0^4\lambda_1^2.
\end{align}

Next, we consider the contribution
\begin{align}\label{eq:cont111}
    \bra{111}_{\alpha\beta\gamma} \eta_{2\rightarrow 2}^{\otimes 3} \ketbra{\psi}^{\otimes 3} \ket{111}_{\alpha\beta\gamma} = 
    \begin{array}{c}
        \\
        \\
        \bra{\psi}^a_{b'} \\
        \otimes \\
        \bra{\psi}^b_{c'} \\
        \otimes \\
        \bra{\psi}^c_{a'}
    \end{array}
    \begin{pmatrix}
        \ketbrau1111_{a}^{a'} \\
        \otimes \\
        \ketbrau1111_{b}^{b'} \\
        \otimes \\
        \ketbrau1111_{c}^{c'}
    \end{pmatrix}
    \begin{array}{c}
        \\
        \\
        \ket{\psi}^a_{b'} \\
        \otimes \\
        \ket{\psi}^b_{c'} \\
        \otimes \\
        \ket{\psi}^c_{a'}
    \end{array}.
\end{align}
The corresponding trivial graph consists of only one node connected to itself, yielding the contribution of 
\begin{align}\label{eq:cont111_b}
    \bra{111}_{\alpha\beta\gamma} \eta_{2\rightarrow 2}^{\otimes 3} \ketbra{\psi}^{\otimes 3} \ket{111}_{\alpha\beta\gamma} = \lambda_1^6.
\end{align}

Likewise, for $\bra{000}_{\alpha\beta\gamma} \eta_{2\rightarrow 2}^{\otimes 3} \ketbra{\psi}^{\otimes 3} \ket{111}_{\alpha\beta\gamma}$ and $\bra{111}_{\alpha\beta\gamma} \eta_{2\rightarrow 2}^{\otimes 3} \ketbra{\psi}^{\otimes 3} \ket{000}_{\alpha\beta\gamma}$, we get a contribution of $\lambda_0^3\lambda_1^3$ each. Summing all contributions to form the fidelity, we get the expression reported in Eq.~\eqref{eq:fidelity_schmidt_22}:
\begin{align}
    F_{\eta_{2\rightarrow2}} &= \frac12(\bra{000}_{\alpha\beta\gamma} \eta_{2\rightarrow 2}^{\otimes 3} \ketbra{\psi}^{\otimes 3} \ket{000}_{\alpha\beta\gamma} + \bra{000}_{\alpha\beta\gamma} \eta_{2\rightarrow 2}^{\otimes 3} \ketbra{\psi}^{\otimes 3} \ket{111}_{\alpha\beta\gamma} \nonumber \\
    &\hphantom{\frac12(}+ \bra{111}_{\alpha\beta\gamma} \eta_{2\rightarrow 2}^{\otimes 3} \ketbra{\psi}^{\otimes 3} \ket{000}_{\alpha\beta\gamma} + \bra{111}_{\alpha\beta\gamma} \eta_{2\rightarrow 2}^{\otimes 3} \ketbra{\psi}^{\otimes 3} \ket{111}_{\alpha\beta\gamma}) \nonumber \\
    &=  \frac12\left(\lambda_0^6+\lambda_1^6+ 3\lambda_0^4\lambda_1^2 + 2\lambda_0^3 \lambda _1^3\right).
\end{align}

This expression can be maximized under the normalization constraint of $\lambda_0^2 + \lambda_1^2 = 1$, yielding a value of $F_{\eta_{2\rightarrow2}} = [5+4\cos(2\pi/7)]/[12+4\cos(2\pi/7)] \approx 0.5170401751642$, achieved by input states with Schmidt coefficients $\lambda_0 \approx 0.913608866$ and $\lambda_1 \approx 0.406594196$.

\subsubsection{Fidelities for $\din > 2$}
\label{app:fiddout2dingr2}

Let us now generalize the argument to higher dimensional input states using the Choi maps displayed in Fig.~\ref{fig:optimaletas} and defined by Eq.~\eqref{eq:etadim2}, which we repeat here for convenience:

\begin{align}\label{eq:etadin2_app}
    \eta_{\din\rightarrow2} = \sum_{r=0}^{\lfloor\frac{\din}{2} \rfloor -1} \left[\tau_{(2r,2r),(2r+1,2r+1)} +  \sum_{k=1}^{\din-2-2r} (\tau_{(2r+k,2r),(2r+k+1,2r+1)} + \tau_{(2r,2r+k),(2r+1,2r+k+1)})\right].
\end{align}
Expanding the threefold tensor product of this Choi map and taking the expectation value with $\ket{\GHZ_3}$, we can again separate four contributions which can be translated into directed graphs that represent the selection rules due to the Schmidt diagonality of the input states $\ket{\psi} = \sum_{i=0}^{d-1} \lambda_i \ket{ii}$. 
We display the corresponding graph for the contribution $\bra{kkk}_{\alpha\beta\gamma} \eta_{\din\rightarrow 2}^{\otimes 3} \ketbra{\psi}^{\otimes 3} \ket{lll}_{\alpha\beta\gamma}$ for fixed $k\leq l \in \{0,1\}$ in Fig.~\ref{fig:graph_din2}, where we only display the edges from the nodes in the first row for clarity. In fact, it is helpful to notice that a node in the $r$'th column can connect to any node in the $r$'th row. This demands that any three-cycle is of the form 
\begin{align}\label{eq:threecycles}
    \ketu{k+r}{k+s}\!\!\tikzmark{a}\!\brau{l+r}{l+s} \rightarrow \ketbrau{k+s}{k+t}{l+s}{l+t}\rightarrow \ketu{k+t}{k+r}\!\tikzmark{b}\!\!\brau{l+t}{l+r}
    \tikz[overlay,remember picture]{\draw[->,square arrow] (b.south) to (a.south);}
\end{align}
for suitable choices of $r,s,t$.

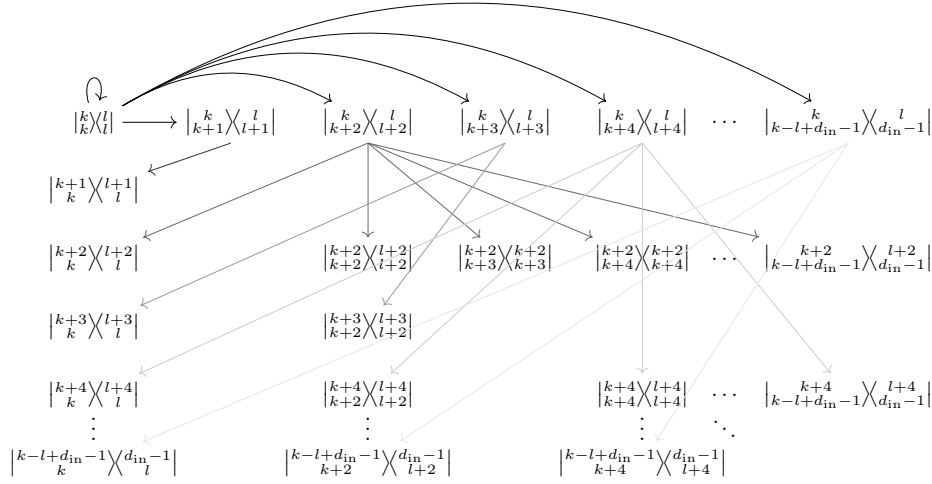
\begin{figure}[t]
    \centering
    \resizebox{0.7\columnwidth}{!}{
    \begin{tikzpicture}[->]
        \node (0000) at (0,0) {$\ketbrau{k}kll$};
        \node (0101) at (2,0) {$\ketbrau{k}{k+1}{l}{l+1}$};
        \node (0202) at (4,0) {$\ketbrau{k}{k+2}{l}{l+2}$};
        \node (0303) at (6,0) {$\ketbrau{k}{k+3}{l}{l+3}$};
        \node (0404) at (8,0) {$\ketbrau{k}{k+4}{l}{l+4}$};
        \node (0ddd0ddd) at (9.2,0) {$\ldots$};
        \node (0dm10dm1) at (11,0) {$\ketbrau{k}{k-l+\din-1}{l}{\din-1}$};

        \node (1010) at (0,-1) {$\ketbrau{k+1}{k}{l+1}{l}$};

        \node (2020) at (0,-2) {$\ketbrau{k+2}{k}{l+2}{l}$};
        \node (2222) at (4,-2) {$\ketbrau{k+2}{k+2}{l+2}{l+2}$};
        \node (2323) at (6,-2) {$\ketbrau{k+2}{k+3}{k+2}{k+3}$};
        \node (2424) at (8,-2) {$\ketbrau{k+2}{k+4}{k+2}{k+4}$};
        \node (2ddd2ddd) at (9.2,-2) {$\ldots$};
        \node (2dm12dm1) at (11,-2) {$\ketbrau{k+2}{k-l+\din-1}{l+2}{\din-1}$};

        \node (3030) at (0,-3) {$\ketbrau{k+3}{k}{l+3}{l}$};
        \node (3232) at (4,-3) {$\ketbrau{k+3}{k+2}{l+3}{l+2}$};

        \node (4040) at (0,-4) {$\ketbrau{k+4}{k}{l+4}{l}$};
        \node (4242) at (4,-4) {$\ketbrau{k+4}{k+2}{l+4}{l+2}$};
        \node (4444) at (8,-4) {$\ketbrau{k+4}{k+4}{l+4}{l+4}$};
        \node (4ddd4ddd) at (9.2,-4) {$\ldots$};
sj        \node (4dm14dm1) at (11,-4) {$\ketbrau{k+4}{k-l+\din-1}{l+4}{\din-1}$};

        \node (5050) at (0,-4.4) {$\vdots$};
        \node (5252) at (4,-4.4) {$\vdots$};
        \node (5454) at (8,-4.4) {$\vdots$};
        \node (5555) at (9.2,-4.4) {$\ddots$};

        \node (dm10dm10) at (0,-5) {$\ketbrau{k-l+\din-1}{k}{\din-1}{l}$};
        \node (dm12dm12) at (4,-5) {$\ketbrau{k-l+\din-1}{k+2}{\din-1}{l+2}$};
        \node (dm14dm14) at (8,-5) {$\ketbrau{k-l+\din-1}{k+4}{\din-1}{l+4}$};

        \path (0000) edge [loop above] (0000);
        
        \path (0000) edge (0101);
        \path (0000) edge [bend left] (0202);
        \path (0000) edge [bend left] (0303);
        \path (0000) edge [bend left] (0404);
        \path (0000) edge [bend left] (0dm10dm1);

        \path [color=black!80] (0101.south) edge (1010);
        
        \path [color=black!60] (0202.south) edge (2020);
        \path [color=black!60] (0202.south) edge (2222);
        \path [color=black!60] (0202.south) edge (2323);
        \path [color=black!60] (0202.south) edge (2424);
        \path [color=black!60] (0202.south) edge (2dm12dm1);

        \path [color=black!40] (0303.south) edge (3030);
        \path [color=black!40] (0303.south) edge (3232);

        \path [color=black!20] (0404.south) edge (4040);
        \path [color=black!20] (0404.south) edge (4242);
        \path [color=black!20] (0404.south) edge (4444);
        \path [color=black!20] (0404.south) edge (4dm14dm1);

        \path [color=black!10] (0dm10dm1.south) edge (dm10dm10);
        \path [color=black!10] (0dm10dm1.south) edge (dm12dm12);
        \path [color=black!10] (0dm10dm1.south) edge (dm14dm14);
    \end{tikzpicture}
    }

    \caption{Part of the directed graph used to determine the $\ket{\GHZ_3}$-fidelity contribution  $\bra{kkk}_{\alpha\beta\gamma} \eta_{\din\rightarrow 2}^{\otimes 3} \ketbra{\psi}^{\otimes 3} \ket{lll}_{\alpha\beta\gamma}$ using the Choi matrices from Eq.~\eqref{eq:etadin2_app} for fixed $k\leq l \in \{0,1\}$. The contribution is given by all three-cycles in this graph. For clarity, we only display the edges from the first line of nodes. In general, a node in the $r$'th column is connected to each node in the $r$'th row. The form of three-cycles in a graph with this structure is given by Eq.~\eqref{eq:threecycles}. }
    \label{fig:graph_din2}
\end{figure}

Carefully counting all such cycles that connect the nodes in the last column and row, i.e., where one of the entries equals $\din-1$, allows us to relate the corresponding fidelity contributions for dimension $\din$ to that of $\din-1$. Doing so for all choices of $k$ and $l$, we obtain:
\begin{theorem}\label{thm:fidelitydin2}
    For symmetric pure input states $\ket{\psi}_{ab'} = \ket{\psi}_{bc'} = \ket{\psi}_{ca'} = \sum_{i=0}^{\din-1} \lambda_i \ket{ii}$ and local channels $\mathcal{E}_A = \mathcal{E}_B = \mathcal{E}_C$ represented by the Choi matrix $\eta_{\din\rightarrow 2}$ in Eq.~\eqref{eq:etadim2}, the fidelity with $\ket{\GHZ_3}$ is given by the recursive rule
    \begin{align}
        F_{\eta_{\din\rightarrow 2}} = F_{\eta_{\din-1\rightarrow 2}}  + \frac12 \begin{dcases}
   \lambda_{\din-1}^6 + 
    2\lambda_{\din-1}^3 \lambda_{\din-2}^3 + 
    3 \lambda_{\din - 1}^2 \Big(\!\!\!\sum_{r,s=0,2,\ldots}^{\din-2} \lambda_r^2 \lambda_s^2 + \sum_{r=1,3,\ldots}^{\din-3} \sum_{s=1,3,\ldots}^{\din-1} \lambda_r^2 \lambda_s^2\Big) + & \\
    \hfill + 6 \lambda_{\din-1}\lambda_{\din-2} 
      \sum_{r=1,3,\ldots}^{\din-3} \sum_{s=1,3,\ldots}^{\din-1} \lambda_r\lambda_{r-1}\lambda_s\lambda_{s-1}
    & \text{if }\din \text{ even,} \\
    \lambda_{\din-1}^6 + 
    \phantom{2\lambda_{\din-1}^3 \lambda_{\din-2}^3 +} 
    3 \lambda_{\din - 1}^2 \Big(\!\!\sum_{r=0,2,\ldots}^{\din-3} \sum_{s=0,2,\ldots}^{\din-1}\lambda_r^2 \lambda_s^2 + \sum_{r,s=1,3,\ldots}^{\din-2}  \lambda_r^2 \lambda_s^2\Big) + & \\
    \hfill + 6 \lambda_{\din-1}\lambda_{\din-2} 
      \sum_{r,s=1,3,\ldots}^{\din-2} \lambda_r\lambda_{r-1}\lambda_s\lambda_{s-1}
    & \text{if }\din \text{ odd}
    \end{dcases}
    \end{align}
    and the fidelity $F_{\eta_{2\rightarrow 2}}$ is given by Eq.~\eqref{eq:fidelity_schmidt_22}.
\end{theorem}
Maximizing these polynomials under the normalization condition of the Schmidt coefficients $\lambda$, we obtain the fidelity values and maximizing squared Schmidt coefficients given in Table~\ref{tab:fidelities_din2}.

\begin{table}[t]
    \centering
    \begin{tabular}{c|c|l}
    $\din$ & max. $F_{\eta_{\din\rightarrow 2}}$ & max. Schmidt coefficients \\
    \hline
        2 & 0.51704017 & (0.91361, 0.40659) \\
        3 & 0.54009112 & (0.82466, 0.52659, 0.20646) \\
        4 & 0.54595881 & (0.79959, 0.54155, 0.24057, 0.09750) \\
        5 & 0.54749297 & (0.79038, 0.54384, 0.24871, 0.12321, 0.04991) \\
        6 & 0.54790031 & (0.78708, 0.54412, 0.25072, 0.13021, 0.06302, 0.02561) \\
        7 & 0.54800887 & (0.78595, 0.54411, 0.25121, 0.13214, 0.06655, 0.03247, 0.01320) \\
        8 & 0.54803772 & (0.78559, 0.54408, 0.25133, 0.13268, 0.06750, 0.03430, 0.01671, 0.00680) \\
        9 & 0.54804537 & (0.78547, 0.54407, 0.25136, 0.13282, 0.06775, 0.03480, 0.01764, 0.00860, 0.00350) \\
        10 & 0.54804739 & (0.78544, 0.54407, 0.25136, 0.13286, 0.06782, 0.03492, 0.01789, 0.00908, 0.00442, 0.00180)
    \end{tabular}
    \caption{Optimal Schmidt coefficients of the bipartite input states and associated obtained $\ket{\GHZ_3}$-fidelities for different input dimensions $\din$ when using the construction based on the Choi matrices in Eq.~\eqref{eq:etadin2_app}. }
    \label{tab:fidelities_din2}
\end{table}

Finally, let us explicitly state the output of the constructed channels, i.e., the state that achieves the reported fidelities for $\din = 10$ up to four digits of precision:
\begin{align}\label{eq:rhod2}
\rho_{\alpha\beta\gamma}^{\din=10} &= \left(\begin{matrix}0.6858 & 0 & 0 & 0 & 0 & 0 & 0 & 0.1805\\0 & 0.0862 & 0 & 0 & 0 & 0 & 0 & 0\\0 & 0 & 0.0862 & 0 & 0 & 0 & 0 & 0\\0 & 0 & 0 & 0.0021 & 0 & 0 & 0 & 0\\0 & 0 & 0 & 0 & 0.0862 & 0 & 0 & 0\\0 & 0 & 0 & 0 & 0 & 0.0021 & 0 & 0\\0 & 0 & 0 & 0 & 0 & 0 & 0.0021 & 0\\0.1805 & 0 & 0 & 0 & 0 & 0 & 0 & 0.0493\end{matrix}\right)
\end{align}

\subsection{Constructions for odd $\dout > 2$}
\label{sec:constd3}

\begin{figure}[t]
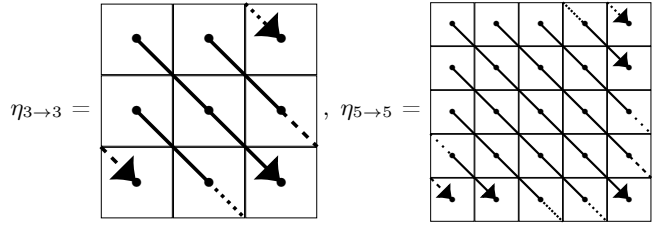

    \centering
    \begin{subfloat}
    
        $\eta_{3\rightarrow3}$ = \raisebox{1cm}{\resizebox{3cm}{!}{\etathreethree}} $\!\!\!,\;\eta_{5\rightarrow5}$ = \raisebox{1.2cm}{\resizebox{3cm}{!}{\etafivefive}}
    \end{subfloat}
    
    \caption{The grid representation of the Choi matrices $\eta_{\dout \rightarrow\dout}$ for $\dout = 3$ and $\dout = 5$ of the grid channels defined in Eq.~\eqref{eq:etadimodd}. Note that this construction only works for odd dimensions and yields the fidelities with $\ket{\GHZ_{3,\dout}}$ reported in Eq.~\eqref{eq:fid_doutdoutmax}.}
    \label{fig:optimaletasodd}
\end{figure}

Let us shift our attention to higher-dimensional outputs of three-partite networks. For a general output dimension of $\dout$, an output fidelity of more than $1/{\dout}$ with $\ket{\GHZ_{3,\dout}}$ indicates genuine multipartite entanglement \cite{bourennane2004experimental}. 

We first consider the case of odd $\dout$ and the same input dimension, as numerical search indicates that curiously, in contrast to the case of $\dout = 2$, in case of $\dout = 3$, the fidelity seems to saturate as soon as $\din = \dout$. For this case, we can produce probably highest fidelities by choosing the grid-type Choi matrices presented in Fig.~\ref{fig:optimaletasodd}, and for general odd dimension $\dout$ given explicitly by
\begin{align}\label{eq:etadimodd}
    \eta_{\dout\rightarrow\dout} &= \tau_{(0,0),(1,1),\ldots,(\dout-1,\dout-1)} + \nonumber \\
    &\phantom{=}\sum_{r=1}^{\frac{\dout-1}2} (\tau_{(0,r),(1,r\oplus 1),\ldots(\dout-1,r\oplus(\dout-1))} + \nonumber\\
    &\phantom{=====}\tau_{(r,0),(r\oplus 1,1),\ldots(r\oplus(\dout-1),\dout-1)}),
\end{align}
where $\oplus$ denotes addition modulo $\dout$.

In the subsequent Appendix~\ref{app:fidodd}, we calculate the fidelity (as a function of the Schmidt coefficients of the input states) for this construction by counting again cycles in an induced graph, yielding
\begin{align}\label{eq:fiddoutdout_schmidt}
    F_{\eta_{\dout\rightarrow\dout}} =     \frac1{\dout}\sum_{k,l=0}^{\dout-1}\left(\lambda_k^3\lambda_l^3 + 3\lambda_k^2\lambda_l^2\sum_{i=1}^{\frac{\dout-1}2} \lambda_{k\oplus i}\lambda_{l\oplus i}\right).
\end{align}
In contrast to $\dout=2$, in this case the optimal values for the Schmidt coefficients is given by $\lambda_i=1/\sqrt{\dout}$, i.e., maximally entangled input states. From the calculation in Appendix~\ref{app:fidodd}, we get an achievable fidelity of 
\begin{align}\label{eq:fid_doutdoutmax}
    F_{\eta_{\dout\rightarrow\dout}} = \frac{3\dout-1}{2\dout^2} = \frac{1}{\dout}\left(\frac32 - \frac1{2\dout}\right)
\end{align}
by using the channels in Eq.~\eqref{eq:etadimodd} and maximally entangled bipartite inputs of dimension $\din = \dout$.

These fidelities numbers are matched by numerical seesaw search for $\dout = 3$. Interestingly, for $\dout=5$, there seem to exist better channels which slightly exceed the fidelity of $0.28$ that is reached by our construction.

Note that the achieved fidelity exceeds the bound of $1/\dout$ (and approaches a value of $3/(2\dout)$ for large dimensions), therefore demonstrating genuine multipartite entanglement for all dimensions. Interestingly, choosing $\lambda_0 = \sqrt{1-\epsilon^2}, \lambda_1 = \epsilon$ and $\lambda_{i>1} = 0$ yields already a fidelity value exceeding the GME bound, thus, even very weakly entangled input states can be mapped to a GME state.

While this construction does not work for even output dimensions $\dout$, we can lift the construction in dimension $\dout-1$ to dimension $\dout$ by simply dismissing the extra dimension. Denoting the constructed $\dout-1$ dimensional network state by $\rho$, we get 
\begin{align}
    \braket{\GHZ_{3,\dout}|\rho|\GHZ_{3,\dout}} = \frac{\dout-1}{\dout}\braket{\GHZ_{3,\dout-1}|\rho|\GHZ_{3,\dout-1}} = \frac{1}{\dout}\left(\frac32 - \frac1{2(\dout-1)}\right) 
\end{align}
for even $\dout$.
    
Together with the value presented in Eq.~\eqref{eq:fid_doutdoutmax}, this proves the lower bound reported in Result~\ref{res:GHZ3dbound} in the main text.

\subsection{General fidelity expression for $\dout > 2$}
\label{app:fidodd}

Here we derive Eq.~\eqref{eq:fiddoutdout_schmidt} from the previous section by calculating the reachable fidelity with $\ket{\GHZ_{3,\dout}}$ using the construction specified in Fig.~\ref{fig:optimaletasodd} and Eq.~\eqref{eq:etadimodd}, which reads
\begin{align}\label{eq:etadimodd_app}
    \eta_{\dout\rightarrow\dout} &= \tau_{(0,0),(1,1),\ldots,(\dout-1,\dout-1)} + \sum_{r=1}^{\frac{\dout-1}2} (\tau_{(0,r),(1,r\oplus 1),\ldots(\dout-1,r\oplus(\dout-1))} + \tau_{(r,0),(r\oplus 1,1),\ldots(r\oplus(\dout-1),\dout-1)}).
\end{align} 
Just as in the case of the analysis of $\dout = 2$, we can split the fidelity with $\ket{\GHZ_{3,\dout}}$ into $\dout^2$ terms of the form $\bra{kkk}_{\alpha\beta\gamma} \eta_{\dout\rightarrow \dout}^{\otimes 3} \ketbra{\psi}^{\otimes 3} \ket{lll}_{\alpha\beta\gamma}$
for $k,l = 0, \ldots, \dout-1$ and some input state $\ket{\psi} = \sum_{i = 0}^{\dout-1} \lambda_i \ket{ii}$. 

Expanding the contribution for some fixed choice of $k$ and $l$ leads to the graph shown in Fig.~\ref{fig:graph_doutdout}. Its three-cycles are of the form 
\begin{align}\label{eq:threecycles_odd}
    \ketu{k}k\!\!\tikzmark{a}\!\brau{l}l \rightarrow \ketbrau{k}{k\oplus i}{l}{l\oplus i} \rightarrow \ketu{k\oplus i}{k}\!\tikzmark{b}\!\!\brau{l\oplus i}{l}
    \tikz[overlay,remember picture]{\draw[->,square arrow] (b.south) to (a.south);}
\end{align}
for $i=0,\ldots,\delta \coloneq  (\dout-1)/2$. Note that whenever $i\neq 0$, we have to count the cycle three times as it may start at any of its nodes, and each cycle contributes $\lambda_k^2\lambda_l^2\lambda_{k\oplus i}\lambda_{l\oplus i}$ to the total fidelity.

\begin{figure}[t]
    \centering
    \resizebox{0.7\columnwidth}{!}{
    \begin{tikzpicture}[->]
        \node (0000) at (0,0) {$\ketbrau{k}kll$};
        \node (0101) at (2,0) {$\ketbrau{k}{k\oplus1}{l}{l\oplus1}$};
        \node (0202) at (4,0) {$\ketbrau{k}{k\oplus2}{l}{l\oplus2}$};
        \node (0303) at (6,0) {$\ldots$};
        \node (0404) at (8,0) {$\ketbrau{k}{k\oplus\delta}{l}{l\oplus\delta}$};
        
        \node (1010) at (0,-1) {$\ketbrau{k\oplus1}{k}{l\oplus1}{l}$};

        \node (2020) at (0,-2) {$\ketbrau{k\oplus2}{k}{l\oplus2}{l}$};

        \node (3030) at (0,-3) {$\vdots$};

        \node (4040) at (0,-4) {$\ketbrau{k\oplus\delta}{k}{l\oplus\delta}{l}$};

        \path (0000) edge [loop above] (0000);
        
        \path (0000) edge (0101);
        \path (0000) edge [bend left] (0202);
        
        \path (0000) edge [bend left] (0404);

        \path (0101) edge  (1010);
        \path (0202) edge  (2020);

        \path (0404) edge  (4040);

        \path [color=black!70] (1010) edge  (0000);
        \path [color=black!70] (1010.east) edge   (0101.south);
        \path [color=black!70] (1010.east) edge   (0202.south);
        \path [color=black!70] (1010.east) edge   (0404.south);

        \path [color=black!50] (2020.west) edge [bend left] (0000.west);
        \path [color=black!50] (2020.east) edge   (0101.south);
        \path [color=black!50] (2020.east) edge   (0202.south);
        \path [color=black!50] (2020.east) edge   (0404.south);

        \path [color=black!30] (4040.west) edge [bend left] (0000.west);
        \path [color=black!30] (4040.east) edge   (0101.south);
        \path [color=black!30] (4040.east) edge   (0202.south);
        \path [color=black!30] (4040.east) edge   (0404.south);

    \end{tikzpicture}
    }

    \caption{The directed graph used to determine the $\ket{\GHZ_{3,\dout}}$-fidelity of the construction using Choi matrices in Eq.~\eqref{eq:etadimodd_app} for odd $\dout$, where $\delta = (\dout-1)/2$ and $k,l = 0, \ldots, \dout-1$. As before, the fidelity is determined by three-cycles in this graph, which are of the form in Eq.~\eqref{eq:threecycles_odd}.}
    \label{fig:graph_doutdout}
\end{figure}
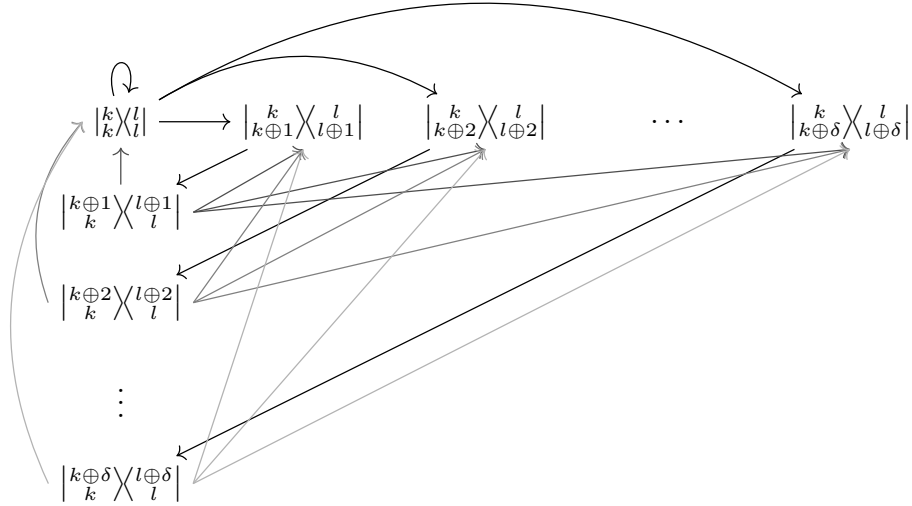

Summing the contributions of all three-cycles for all $k$ and $l$ and taking into account the normalization of the GHZ state, we obtain the result reported in Eq.~\eqref{eq:fiddoutdout_schmidt}:
\begin{align}
    F_{\eta_{\dout\rightarrow\dout}} =     \frac1{\dout}\sum_{k,l=0}^{\dout-1}\left(\lambda_k^3\lambda_l^3 + 3\lambda_k^2\lambda_l^2\sum_{i=1}^{\frac{\dout-1}2} \lambda_{k\oplus i}\lambda_{l\oplus i}\right).
\end{align}
In contrast to the case of $\dout = 2$, in this case, the optimum seems to be achieved for $\lambda_i = 1/\sqrt{\dout}$ for all $i$, for which the fidelity reaches the value reported in Eq.~\eqref{eq:fid_doutdoutmax} of 
\begin{align}
    F_{\eta_{\dout \rightarrow \dout}} = \frac{3\dout-1}{2\dout^2}.
\end{align}

\section{Fidelity bounds in larger networks from bounds in triangular networks}
\label{app:losr4to3}

\subsection{Reducing fidelity bounds in  large networks to smaller ones}

In order to find fidelity bounds for states in $\LOSR$ with $N$ nodes, bipartite sources, it is often possible to reduce this task to finding bounds in smaller networks. Intuitively, the idea is to group nodes in the $N$-partite network and allow them to perform  a unitary post-processing, followed by tracing out all but one node per group. This procedure can only lead to higher fidelities with the target state while at the same time producing a network state of fewer parties. We formally prove this argument in the following lemma.

\begin{lemma}[State extraction]
    \label{lem:stateextraction}
    Let $\ket{\psi_N}$ be an $N$-partite quantum state of local dimension $d$. Let the $N$ parties be partitioned into $n\geq 3$ sets $A_1,\ldots,A_n$.
    
    If there exist unitary operators $U_{A_i}$ acting only on the  parties in $A_i$ for $i=1,\ldots, n$ such that
    \begin{align} 
        U_{A_1} \otimes \ldots \otimes U_{A_n} \ket{\psi_N} = \ket{\varphi_n}_{a_1,\ldots a_n} \otimes \ket{\tau}_{R},
    \end{align}
    where $\ket{\varphi_n}$ is an $n$-partite state of parties $a_1 \in A_1, \ldots, a_n \in A_n$, and where $\ket{\tau}_R$ is a residual state on the parties in $R = \set{1, \dots, N} \setminus \set{a_1, \ldots a_n}$, then the maximal fidelity of $\LOSR$-states in $N$-partite networks with $\ket{\psi_N}$ cannot be larger than that of $\LOSR$-states in $n$-partite networks with $\ket{\varphi_n}$, i.e.,
    \begin{align}
        F_{\LOSR}(\ket{\psi_N}) \leq F_{\LOSR}(\ket{\varphi_n}).
    \end{align}
\end{lemma}
\begin{proof}
    For clarity, we prove the statement for $N$ parties, partitioned into three sets  $A_1 = A$, $A_2 = B$, $A_3=C$. The argument can then directly be extended to  larger $n$.
    
    To prove the statement, we consider the larger class of $N$-partite LOSR networks with bipartite sources between all parties, where now the parties in $A$, $B$ and $C$ are allowed to apply set-local unitaries $U_A$, $U_B$ and $U_C$, respectively, to their outputs (see Fig.~\ref{fig:graphstate_trip}). Let us denote this class by $\underline{\LOSR}_d^N$. Clearly, the maximal fidelity achievable in the usual $N$-partite LOSR network cannot exceed the one obtainable in the $\underline{\LOSR}$ network. Let $\rho$ denote a state in $\LOSR$. Then its fidelity with $\ket{\psi_N}$ is given by
    \begin{align}
        F' &= \braket{\psi_N | \rho | \psi_N}\\
        &\leq \max_{\rho' \in \underline{\LOSR}}\braket{\psi_N | \rho' | \psi_N}\\
        &= \max_{\rho' \in \underline{\LOSR}} \braket{\psi_N | \big( U_A^\dagger \otimes U_B^\dagger \otimes U_C^\dagger \big) \rho' \big( U_A\otimes U_B\otimes U_C \big) | \psi_N}\\
        &= \max_{\rho' \in \underline{\LOSR}} \bra{\varphi_3} \otimes \bra{\tau} \rho' \ket{\varphi_3}\otimes\ket{\tau} \\
        &\leq \max_{\rho' \in \underline{\LOSR}} \braket{\varphi_3 | \trace_{R}(\rho') | \varphi_3}\\
        &\leq F_{\LOSR}(\ket{\varphi_3}).        
    \end{align}%
\begin{figure}[t]
    \centering
    \includegraphics[width=0.5\linewidth]{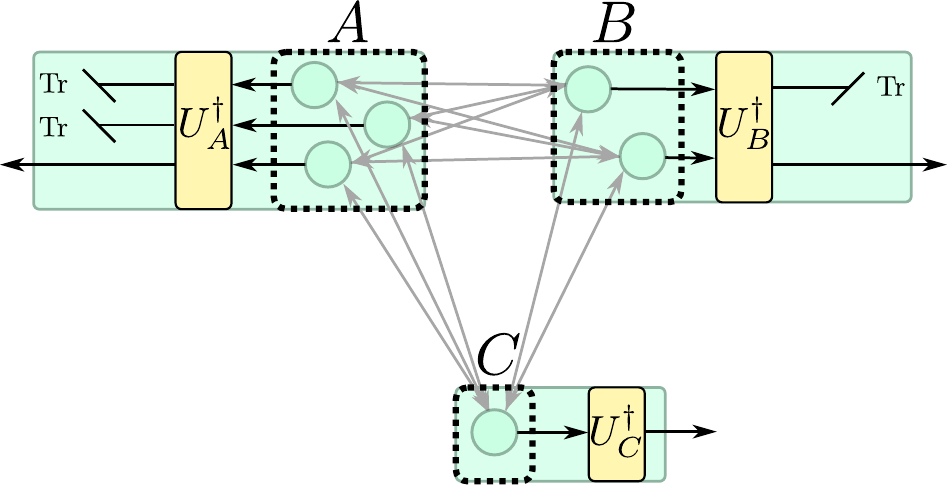}
    \caption{Visualization of the proof of Lemma~\ref{lem:stateextraction} for the tripartition $A=\{1,2,3\}$, $B=\{4,5\}$, $C=\{6\}$. 
    The $6$-partite network can be used to create a tripartite network by applying local unitaries $U_A$, $U_B$ and $U_C$ and then tracing out all but one party per set. If a $6$-partite state $\ket{\psi_6}$ can be transformed using the unitaries $U_A$, $U_B$ and $U_C$ into a $3$-partite state $\ket{\varphi_3}$ between some parties in $A$, $B$ and $C$ (and some residual state), then the fidelity with $\ket{\psi_6}$ in the full network is upper bounded by the fidelity with $\ket{\varphi_3}$ in the tripartite network.
    }
    \label{fig:graphstate_trip}
\end{figure}%
The second-last inequality follows from the fact that the fidelity between two states can only increase under the partial trace.
Finally, the last inequality follows from the fact that the state $\trace_{R}(\rho')$ is preparable in a tripartite LOSR network. To see this, we combine the nodes in $A$ (and likewise for $B$ and $C$) into a single node. The channels of these nodes are then given by first applying the original parties' channels, followed by the unitary post-processing and the subsequent partial trace. Additionally, the bipartite sources connecting the parties in $A$ and $B$ are combined into a single (larger-dimensional) bipartite source between $A$ and $B$, and likewise for the other two pairs.
\end{proof}

Let us stress that this proof also bounds the fidelities for a larger class of networks, where instead of having bipartite sources between all parties, we allow even for multipartite sources connecting the parties in each pair of sets $A_i$ and $A_j$ for $1\leq i < j \leq n$. 

An immediate application of Lemma~\ref{lem:stateextraction} concerns the fidelity LOSR-states with GHZ states:

\begin{corollary} \label{cor:ghzfidelity}
    For all $3\leq n < N$, 
    \begin{align}
        F_{\LOSR}(\ket{\GHZ_{N,d}}) \leq F_{\LOSR}(\ket{\GHZ_{n,d}})
    \end{align} 
\end{corollary}
\begin{proof}
    We use Lemma~\ref{lem:stateextraction} and choose $\ket{\psi_N} = \ket{\GHZ_{N,d}}$ and $\ket{\varphi_n} = \ket{\GHZ_{N,d}}$.
    We further choose $A_1 = \{1\}, A_2=\{ 2\}, \ldots, A_{n-1} = \{n-1\}$ and $A_n = \{n, n+1, \ldots, N\}$, such that $\vert A_n \vert = N-n+1$.
    
    Finally, we set $U_{i} = \one$ for $1\leq i \leq n-1$ and $U_{n}$ such that $U_{n} \ket{i}^{\otimes(N-n+1)} = \ket{i}\otimes \ket{0}^{\otimes (N-n)}$ for $0\leq i < d$.
    This yields
    \begin{align}
        U_1\otimes \ldots \otimes U_n \ket{\GHZ_{N,d}} = \ket{\GHZ_{n,d}} \otimes \ket{0}^{\otimes(N-n)}.
    \end{align}
    The claim then follows from Lemma~\ref{lem:stateextraction}.
\end{proof}

\subsection{Bounding fidelities of graph states} \label{app:graphstates}

We now shift our attention to graph states \cite{schlingemann2001quantum, raussendorf2003measurement, hein2004multiparty}. 
This family of multipartite states has found ubiquitous applications in quantum information theory and quantum computing and includes prominent examples of states like the multipartite GHZ state, as well as cluster states, which are resource states for measurement based quantum computation \cite{raussendorf2003measurement}. 
Furthermore, each stabilizer state is locally equivalent to a graph state \cite{van2004graphical}, which makes the results of this section applicable to them as well.

Previous results have shown that no connected graph state can be prepared
in LOSR networks with arbitrary bipartite sources with fidelities exceeding $9/10$  \cite{wang2024quantum, makuta2023no}. Here, we improve this result for all graph states.
For a comprehensive introduction to the graph state formalism, see, for example, Ref.~\cite{hein2004multiparty}. Here, we just state the basic facts required for our argument.

\subsubsection{Graph state preliminaries}

Let us now consider the case of graph states, for which we start by stating the basic definitions and proving some useful Lemmas.

A graph state $\ket{G}$ of $N$ qubits is defined via a graph $G=(V,E)$ of $N$ vertices $V=\{1,\ldots,N \}$ and edges $E \subset V \times V$ between them via
\begin{align}
    \ket{G} = \prod_{(i,j) \in E} \CZ_{i,j} \ket{+}^{\otimes N},
\end{align}
where $\ket{+} = (\ket{0} + \ket{1})/\sqrt{2}$ and
\begin{align}
    \CZ_{i,j} = \begin{pmatrix}
    1 & 0 & 0 & 0\\ 0 & 1 & 0 & 0\\ 0 & 0 & 1 & 0\\ 0 & 0 & 0 & -1
    \end{pmatrix}_{i,j} \otimes \one_{\{1,\ldots,N\}\setminus \{i,j\}}
\end{align}
denotes the self-inverse controlled-$Z$ operation on qubits $i$ and $j$. Note that due to the one-to-one correspondence of graphs and graph states, we sometimes use $G$ and $\ket{G}$ interchangeably. 

Of general interest is the question of when two graph states $\ket{G}$ and $\ket{G'}$ are locally equivalent, i.e., whether there exists a (local) unitary transformation $U = U_1 \otimes U_2 \otimes \ldots \otimes U_N$, such that $U\ket{G} = \ket{G'}$. In many cases, this can be checked by considering local complementations of the associated graph $G$. Denoting the neighborhood of an vertex $v$ in $G$ by $\NN(v, G) = \{i \in V\,:\, (i,v) \in E\}$, the local complementation $\tau_v$ with respect to $v$ can be graphically interpreted as inverting every edge in the neighborhood of $v$. If two vertices in the neighborhood of $v$ are connected in $G$, this edge will be removed in $\tau_v(G)$, and vice versa. Mathematically, a local complementation $\tau_v$ w.r.t.~a vertex $v\in V$  transforms the graph $G = (V,E)$ into the graph $G' = (V,E')$, where $E' = (E \setminus \{(i,j) \,:\, i,j \in \NN(v , G) \wedge (i,j) \in E\}) \cup \{(i,j) \,:\, i,j \in \NN(v , G) \wedge (i,j) \notin E\}$.

We note here that applying the controlled-$Z$ operation $\CZ_{i,j}$ to the graph state $\ket{G}$ gives a graph state $\ket{G'}$, where $G'$ has the same vertices and edges as $G$, except for the edge between the vertices $i$ and $j$. If the edge $(i,j)$ is present in $G$, it will be absent in $G'$ and vice versa. Thus, the action of a local complementation on a graph state $\ket{G}$ can be written equivalently in terms of controlled-$Z$ operations as
\begin{equation}
    \tau_v(\ket{G}) = \prod_{k, l \in \NN(v, G), k < l} \CZ_{k,l} \ket{G}.
\end{equation}

The relevance of local complementations is then given by the following result:
\begin{theorem}[\cite{van2004efficient, hein2004multiparty}]
    If a sequence of local complementations  $\tau_{v_1}\ldots \tau_{v_k}$ transforms the graph $G$ into $G'$, then the corresponding graph states $\ket{G}$ and $\ket{G'}$ are locally equivalent.
\end{theorem}

In the next section, we make use of a special case of the commutation relation between the local complementation operations and the controlled-$Z$ operations. 
In the following, we will use the notion of removing a vertex $v$ of a graph $G$, which can be seen as deleting all edges to $v$ in the graph $G$. Using the notation \begin{align}\label{eq:def_setCZ}
    \CZ_{M,v} \coloneq  \prod_{j \in M} \CZ_{j,v}
\end{align}
for any subset $M$ of the vertices, the deletion of all edges to the vertex $v$ in the graph $G$ leads to the graph state $\CZ_{\NN(v,G),v} \ket{G}$.
Interestingly, this action commutes with local complementations on vertices unequal to $v$. Since a local complementation $\tau_j$ may only create and delete edges to $v$, it is clear that first deleting all edges to $v$ in the graph $G$ and then applying the local complementation $\tau_j$ has the same effect as deleting all edges to $v$ in the graph $\tau_j(G)$ obtained by applying the local complementation first. 
Since also all controlled-$Z$ operations commute, we arrive at the following lemma:
\begin{lemma}\label{cor:com_rel}
    For any graph state $\ket{G}$, one of its vertices $v$ and a sequence $U$ of controlled-$Z$ gates and local complementations of vertices unequal to $v$, it holds that
    \begin{align}
        U \CZ_{\NN(v,G),v} \ket{G} = \CZ_{\NN(v, UG),v} U \ket{G}.
    \end{align}
\end{lemma}

\begin{figure}
    \centering
    \includegraphics[width=1\linewidth]{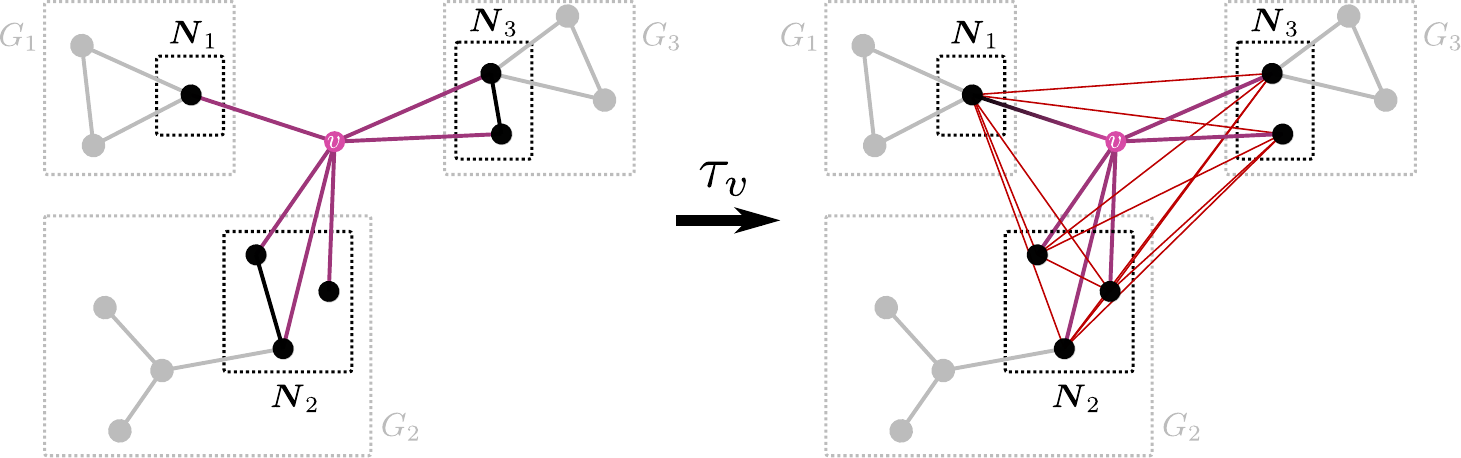}
    \caption{Left: visualization of an articulation point $v$ in a graph which connects the three components $G_1$, $G_2$ and $G_3$. The neighborhood of $v$ inside $G_i$ is denoted by $\NN_i$. Right: the same graph after a local complementation of $v$. Even though some edges within the neighborhood are removed, all vertices remain connected after $v$ is removed. Thus, $v$ is not an articulation point anymore.  }
    \label{fig:articulation}
\end{figure}

As a final ingredient, we need the following observation regarding articulation points of a graph, which are vertices such that their removal increases the number of components of the graph:
\begin{lemma} \label{lem:articulation}
    Let $G = (V,E)$ be a connected graph and $v \in V$ be an articulation point of $G$. Then $v$ is not an articulation point in the graph $\tau_v(G)$.
\end{lemma}
\begin{proof}
    We have to show that there exists a path between any pair of vertices $x\neq v,y\neq v$ in $G' \coloneq  \tau_v(G) \setminus \{v\}$.
    
    As $v$ is an articulation point, its removal splits $G$ into $k\geq 2$ components $G_1,\ldots, G_k$ which are not connected to each other. Denote the neighborhood of $v$ in $G$ as before by $\NN(v,G)$ and define $\NN_j \coloneq  G_j \cap \NN(v,G)$ for $1\leq j\leq k$ as the part of $G_j$ that is connected to $v$ (see Fig.~\ref{fig:articulation}). Note that in $G'$, for all $1\leq i < j\leq k$, all vertices in $\NN_i$ are connected to all vertices in $\NN_j$ due to the local complementation. 

    Now we consider different cases for $x$ and $y$.

    \begin{description}
        \item[Case 1 ($x \in G_i \ni y$)]  Then $x$ and $y$ are connected in $G$ via a path within $G_i$. This path might contain edges within $\NN_i$ which got removed due to the local complementation. However, since all vertices in $\NN_i$ are connected to all vertices in $\NN_j$ for any $j \neq i$, we can replace the missing edge in the path by a detour via some vertex in $\NN_j$. Thus, $x$ and $y$ are still connected.

        \item[Case 2 ($x \in G_i, y \in G_j$)] In this case, $x$ and $y$ would not be connected in $G \setminus \{v\}$, as any path from $x$ to $y$ would contain the sequence $n - v - m$ where $n\in \NN_i$ and $m \in \NN_j$, and possibly edges within $\NN_i$ and edges within $\NN_j$. However, noting that the first vertex in the path that belongs to $\NN_i$ is directly connected to the last vertex in the path that belongs to $\NN_j$ yields a path connecting $x$ and $y$ in $G'$.
    \end{description}   
\end{proof}

\subsubsection{GHZ extraction}

Here, we show that an extraction of a three-partite GHZ state from arbitrary connected graph states is always possible.
\begin{theorem}[Graph state GHZ extraction]\label{thm:ghzextraction}
    For all connected graph states $\ket{G}$ of $N\geq 3$ qubits and any choice of three of its parties $a$, $b$ and $c$, there exists a proper tripartition $A \cup B \cup C = \{1,\ldots,N\}$ with $a \in A$, $b\in B$ and $c\in C$ and set-local unitary operations $U_A$, $U_B$ and $U_C$ acting on $A$, $B$ and $C$, respectively, such that 
    \begin{align}
        U_A \otimes U_B \otimes U_C \ket{G} = \ket{\GHZ_3}_{abc} \otimes \ket{G'},
    \end{align}
    where $\ket{G'}$ is a graph state of $N-3$ vertices.
\end{theorem}

\begin{proof}
    It is sufficient to prove the same statement where we replace $\ket{\GHZ_3}$ by the local unitary equivalent three-partite fully connected graph state $\ket{\triangle}$. 
    We prove the theorem by induction over $N$.
    
    Let $N=3$ and $\ket{G}$ be a connected graph state of three vertices and w.l.o.g.~$a=1$, $b=2$, $c=3$. The only possible choice of sets is given by $A=\{1\}$, $B=\{2\}$, $C=\{3\}$. As $\ket{G}$ is connected, it either already equals $\ket{\triangle}$, or it is missing one edge. By local complementation of the opposite vertex, this edge is created, thereby transforming $\ket{G}$ into $\ket{\triangle}$.

    Now, let the claim be true for some fixed $N-1\geq 3$. We consider the $N$-partite graph $\ket{G_N}$ and three arbitrary vertices $\{a,b,c\}$.
    We choose an arbitrary vertex $v \notin \{a,b,c\}$, such that we can write
    \begin{align}
        \ket{G_N} = \prod_{j \in \NN(v, G_N)} \CZ_{j,v} \ket{G_{N-1}} \otimes \ket{+}_{v},
    \end{align}
    where $\ket{G_{N-1}}$ is a graph state of $N-1$ vertices, including $a$, $b$ and $c$. The goal is now to use the assumption of induction on $\ket{G_{N-1}}$, but it might be that $\ket{G_{N-1}}$ is not connected. This happens if $v$ is an articulation point of the graph $G_N$, in which case we make use of Lemma~\ref{lem:articulation} and first apply a local complementation w.r.t.~$v$. To keep track of this, we define
    \begin{align}
        \ket{\tilde{G}_N} = \begin{cases}
            \tau_v(\ket{G_N}) & \text{if $v$ is an articulation point of $G_N$} \\
            \ket{G_N} & \text{otherwise.}
        \end{cases}
    \end{align}
    and decompose
    \begin{align}\label{eq:Gtilden}
        \ket{\tilde{G}_N} = \prod_{j \in \NN(v, \tilde{G}_N)} \CZ_{j,v} \ket{\tilde{G}_{N-1}} \otimes \ket{+}_{v},
    \end{align}
    now with a certainly connected $N-1$-partite graph state $\ket{\tilde{G}_{N-1}}$.
    By assumption of induction, there exists a tripartition $A \cup B \cup C = \{1,\ldots,N-1\}$ of the vertices of $\ket{\tilde{G}_{N-1}}$ and set-local unitaries $U_A$, $U_B$ and $U_C$, such that
    \begin{align}\label{eq:UABCGeqTriGp}
        U_A \otimes U_B \otimes U_C \ket{\tilde{G}_{N-1}} = \ket{\triangle}_{abc} \otimes \ket{G'},
    \end{align}
    where $\ket{G'}$ is some graph state of $N-4$ vertices.
    Note that we construct $U_A\otimes U_B \otimes U_C$ further below such that it is not only a set-local unitary, but it consists in fact only of products of $\CZ$ gates within the sets $A$, $B$ and $C$, respectively, and local complementations. Thus, combining Eq.~\eqref{eq:UABCGeqTriGp} and Eq.~\eqref{eq:Gtilden} we find 
    \begin{align}
        U_A\otimes U_B\otimes U_C \prod_{j \in \NN(v, \tilde{G}_N)} \CZ_{j,v} \ket{\tilde{G}_N} =  \ket{\triangle}_{abc} \otimes \ket{G'} \otimes \ket{+}_{v}. 
    \end{align}
    Using the definition in Eq.~\eqref{eq:def_setCZ}, we can shorten the left-hand side of this equation and apply Lemma~\ref{cor:com_rel}:
    \begin{align}
        U_A\otimes U_B\otimes U_C \prod_{j \in \NN(v, \tilde{G}_N)} \CZ_{j,v} \ket{\tilde{G}_N} &= U_A\otimes U_B\otimes U_C \CZ_{\NN(v, \tilde{G}_N),v} \ket{\tilde{G}_N}  \\
        &= \CZ_{\NN(v, U_A\otimes U_B\otimes U_C\tilde{G}_N),v} U_A\otimes U_B\otimes U_C \ket{\tilde{G}_N} .
    \end{align}
    Applying the controlled-$Z$ operations on the right-hand side of Eq.~\eqref{eq:Gtilden} again, we get 
    \begin{align}
        U_A\otimes U_B\otimes U_C \ket{\tilde{G}_N} = \CZ_{\NN(v, U_A\otimes U_B\otimes U_C\tilde{G}_N),v} \ket{\triangle}_{abc} \otimes \ket{G'} \otimes \ket{+}_{v}. 
    \end{align}
    Now, depending on the connections of $v$ to $a$, $b$ and $c$ in $U_A\otimes U_B\otimes U_C\ket{\tilde{G}_N}$, we can decide to which of the sets $A$, $B$ or $C$ we add the vertex $v$ and how to modify $U_A$, $U_B$ and $U_C$ such that they stay set-local and disconnect $\ket{\triangle}_{abc}$ from $v$. We consider each case individually. To that end, set 
    $\NN \coloneq  \NN(v, U_A\otimes U_B\otimes U_C\tilde{G}_N)$. 
\begin{figure}
    \centering
    \includegraphics[width=0.7\linewidth]{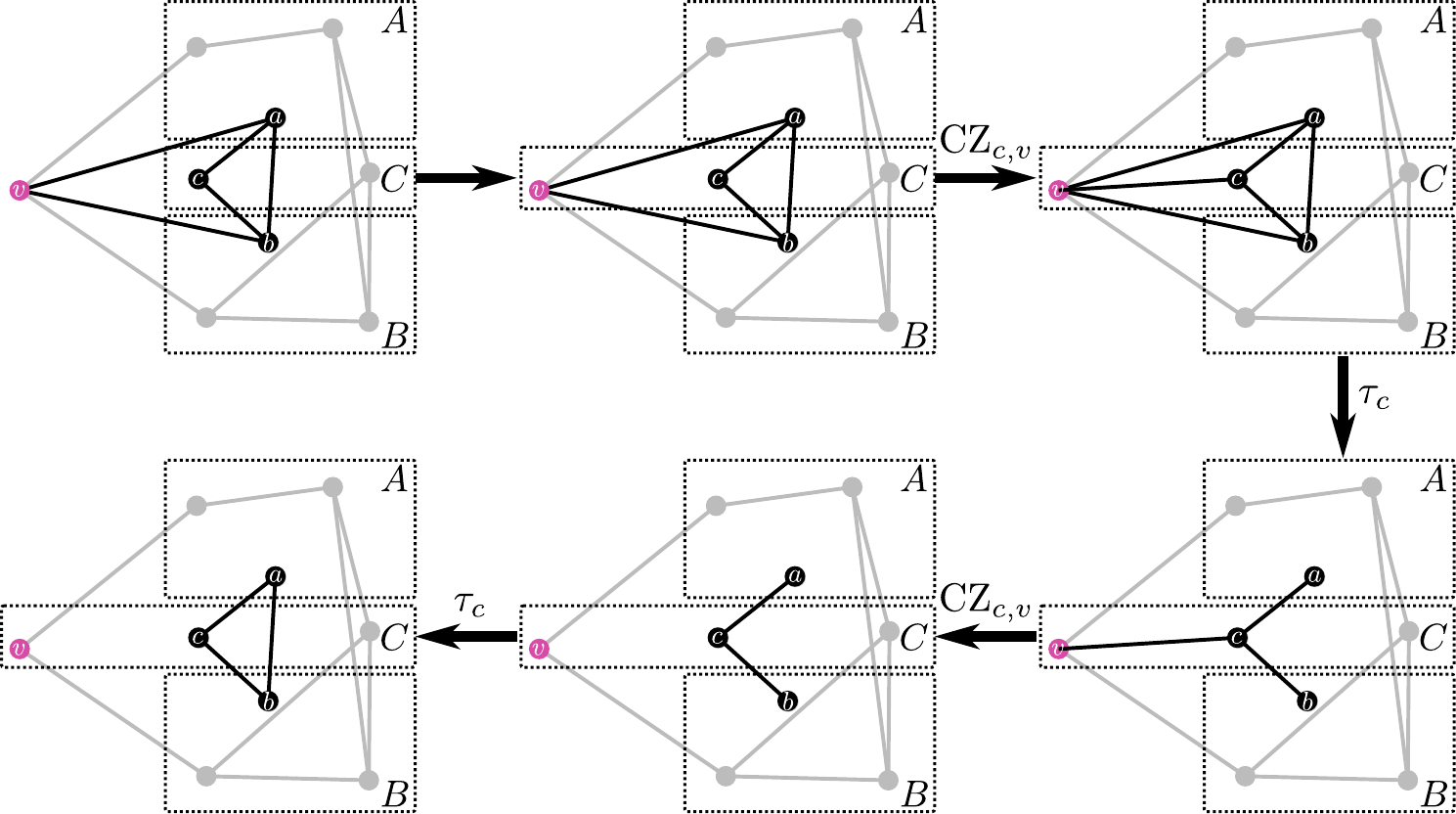}
    \caption{Visualization of case 3 of the proof: If the new vertex $v$ is connected to vertices $a$ and $b$ of the extracted triangle graph state, we add it to set $C$ and perform the $C$-local operations $\CZ_{c,v}$, $\tau_c$, $\CZ_{c,v}$, $\tau_c$ to disconnect the vertices $a,b,c$ from the new vertex $v$.}
    \label{fig:ghzextraction}
\end{figure}
    
    \begin{description}
        \item[Case 1, $a,b,c \notin \NN$]~\\
        In this case, $\CZ_{\NN,v} \ket{\triangle}_{abc} \otimes \ket{G'} \otimes \ket{+}_{v} =  \ket{\triangle}_{abc} \otimes \CZ_{\NN,v} \ket{G'} \otimes \ket{+}_{v}$, where the triangle state is disconnected from the rest. Thus, we do not have to modify the local unitary and we can add the new vertex $v$ to any of the sets $A$, $B$ or $C$ in order to obtain a proper tripartition of the vertices.
        \item[Case 2, $a\in \NN$, $b,c \notin \NN$]~\\
        In this case (and analogously in cases where $b$ or $c$ instead of $a$ are connected to $v$), we commute all controlled-$Z$ gates but $\CZ_{a,v}$ past $\ket{\triangle}_{abc}$:
        \begin{align}
             U_A\otimes U_B\otimes U_C \ket{\tilde{G}_N} = \CZ_{a,v}  \ket{\triangle}_{abc} \otimes \CZ_{\NN\setminus \{a\},v}\ket{G'} \otimes \ket{+}_{v}. 
        \end{align}
        We add the new vertex $v$ to the set $A$ such that we can absorb $\CZ_{a,v}$ into $U_A$ on the left-hand side to yield the claim.
        \item[Case 3, $a,b \in \NN$, $c\notin \NN$] In this case (and analogously in cases where $a$ or $b$ instead of $c$ are not connected to $v$), we write 
        \begin{align}
            U_A\otimes U_B\otimes U_C \ket{\tilde{G}_N} = \CZ_{a,v}\CZ_{b,v}  \ket{\triangle}_{abc} \otimes \CZ_{\NN\setminus \{a,b\},v}\ket{G'} \otimes \ket{+}_{v}. 
        \end{align}
        We add $v$ to $C$ and multiply both sides from the left by the set-local sequence $\tau_c \CZ_{c,v} \tau_c \CZ_{c,v}$ (such that $\CZ_{c,v}$ is applied first). As visualized in Fig.~\ref{fig:ghzextraction}, this sequence disconnects the triangle again from the rest. Thus,  
        \begin{align}
            \tau_c \CZ_{c,v} \tau_c \CZ_{c,v}  U_A\otimes U_B\otimes U_C \ket{\tilde{G}_N} =   \ket{\triangle}_{abc} \otimes \CZ_{\NN\setminus \{a,b\},v}\ket{G'} \otimes \ket{+}_{v}. 
        \end{align}
        
        \item[Case 4, $a,b,c\in \NN$]~\\
        In this case, add $v$ to $C$, and, analogously to case 3, we skip the first application of $\CZ_{c,v}$ and just apply $\tau_c \CZ_{c,v} \tau_c$, yielding the claim.
    \end{description}
    Finally, we absorb the potential local complementation of $v$ from the definition of $\ket{\tilde{G}_N}$ into $U_A \otimes U_B \otimes U_C$ to yield the claim.
    Therefore, in any case we are able to recover the structure of an isolated triangle graph that is locally equivalent to $\ket{\GHZ_3}$.
\end{proof}

We are now in position to put all ingredients together: Knowing from Theorem~\ref{thm:ghzextraction} that we can extract a $\ket{\GHZ_3}$-state from every connected graph state of $3$ or more parties, we can apply Lemma~\ref{lem:stateextraction}, which immediately yields the claim of Result~\ref{res:graphstatebound} from the main text:

\setcounter{thmcbak}{\thethmc}
\setcounter{thmc}{1}
\begin{theorem}
    No connected qubit graph state of $N \geq 3$ parties can be prepared in an $N$-partite LOSR network with bipartite sources with a fidelity larger than $F_{\LOSR}(\ket{\GHZ_3})$.
\end{theorem}
\setcounter{thmc}{\thethmcbak}

Let us stress that the task of GHZ state extraction from graph states with a fixed tripartition of the vertices has been considered before \cite{bravyi2006ghz}. In fact, the number of extractable GHZ states for a fixed tripartition has been shown to be an entanglement measure \cite{linden2002almost} and is equal to the difference of the dimension of the stabilizer group $S$ of a stabilizer state (which is equal to the number of parties $N$) and the dimension of the group that is generated by the union of the local subgroups $S_{AB}$, $S_{AC}$ and $S_{BC}$. Thus, the result of Theorem~\ref{thm:ghzextraction} sheds additional light on the structure of qubit stabilizer groups:
\begin{corollary}
    For every $N$-partite qubit stabilizer group $S$ of size $2^N$ (i.e., dimension $N$) and every choice of three of its parties $a$, $b$ and $c$, there exists a tripartition $A$, $B$, $C$ of the parties with $a\in A$, $b \in B$ and $c \in C$ such that $\dim(\langle S_{AB}, S_{AC}, S_{BC}\rangle) < \dim(S) = N$. 
\end{corollary}

This, in turn, has an interesting consequence in terms of whether a stabilizer state is uniquely determined by its marginals:
\begin{corollary}
    For every $N$-partite qubit stabilizer state $\ket{\psi}$ and every choice of three of its parties $a$, $b$ and $c$, there exists a tripartition $A$, $B$, $C$ of the parties with $a\in A$, $b \in B$ and $c \in C$ such that $\ket{\psi}$ is not uniquely determined by the reduced states $\rho_{AB} = \trace_C(\ket{\psi}\!\bra{\psi})$, $\rho_{AC}$ and $\rho_{BC}$.
\end{corollary}

\section{Analytical upper bounds on the two-dimensional GHZ state fidelity}
\label{app:analyticalbound}

Here, we prove the upper bound of $F_{\LOSR}(\ket{\GHZ_3}) < 0.618$ from Result~\ref{res:GHZ32bound} for tripartite networks with qubit outputs, which improves upon the previously known bound of
$F_{\LOSR} \leq (1+\sqrt{3})/4 < 0.684$~\cite{navascues2020genuine}.
Note that in Refs.~\cite{navascues2020genuine} and \cite{smith2025fully}, the numerical value is incorrectly reported as $0.6803$.

We prove the statement for tripartite network states without shared randomness (i.e, $\text{LO}$ states, which are sometimes called independent triangle (ITN) states). For such states, stricter inequalities are known. Since the fidelity is linear if at least one of the states is pure, the bounds obtained for $\text{LO}$-states will then also be valid for $\LOSR$-states.

We start with an alternative formulation for the fidelity between two quantum states $\rho$ and $\sigma$ \cite{bengtsson2017geometry}:
\begin{align}
\label{eq:fidmeasurement}
    F(\rho,\sigma) = \left(\min_{\{E_k\}_k} \sum_k \sqrt{\trace{[E_k\rho]}}\sqrt{\trace{[E_k\sigma]}}\right)^2,
\end{align}
where the minimization ranges over all positive operator valued measures (POVMs), i.e., sets of positive semidefinite, hermitian operators $E_k$ with $\sum_k E_k = \one$. Thus, for any fixed choice of a POVM $\{E_k\}_k$, the fidelity is bounded by
\begin{align}
    \label{eq:fidelity_proof1}
    F(\rho, \sigma)\leq \left(\sum_{k}\sqrt{\trace{[E_k\rho]}}\sqrt{\trace{[E_k\sigma]}}\right)^2.
\end{align}
We choose $\sigma=\ketbra{\GHZ_3}$ to be the target state, $\rho$ to be preparable in a triangle network without shared randomness and the POVM to be the projective measurement in the computational basis, i.e. $\{\ketbra{ijk}\}_{i,j,k \in \{0,1\}}$. Then, inequality \eqref{eq:fidelity_proof1} reduces to
\begin{align}
\label{eq:fidelity_proof2}
    F(\rho,\sigma)\leq \frac{1}{2}\left(\sqrt{\trace{[\ket{000}\bra{000}\rho]}}+\sqrt{\trace{[\ket{111}\bra{111}\rho]}}  \right)^2.
\end{align}
We write $\ketbra{000} = (\one + Z)^{\otimes 3}/8$ and $\ketbra{111} = (\one - Z)^{\otimes 3}/8$ to rewrite the r.h.s.~of \eqref{eq:fidelity_proof2} in terms of the expectation values of the Pauli-Z  matrices w.r.t.~$\rho$, i.e.,
\begin{align}\label{eq:Fbound_116}
    F(\rho, \sigma)\leq  \frac{1}{16}\Big(\sqrt{1+3z_1+3z_2+z_3}+\sqrt{1-3z_1+3z_2-z_3}\Big)^2,
\end{align}
where $z_1\equiv (\braket{Z_A}+\braket{Z_B}+\braket{Z_C})/3$, $z_2\equiv(\braket{Z_AZ_B}+\braket{Z_AZ_C}+\braket{Z_BZ_C})/3$ and $z_3\equiv\braket{Z_AZ_BZ_C}$.

Before proving the statement, let us establish a fidelity bound involving $z_2$ only.
To that end, notice that we can express the fidelity in terms of the stabilizers of the GHZ state via
\begin{align}
    F(\rho,\sigma) = \frac18(1+3z_2 + \braket{X_AX_BX_C} - \braket{X_AY_BY_C} - \braket{Y_AX_BY_C} - \braket{Y_AY_BX_C}).
\end{align}
Furthermore, we can write the fidelity with the orthogonal state $\ket{\overline{\GHZ}_3} \coloneq  \frac1{\sqrt{2}}(\ket{000}-\ket{111})$ as 
\begin{align}
    0\leq F(\rho,\ketbra{\overline{\GHZ}_3}) = \frac18(1+3z_2 - \braket{X_AX_BX_C} + \braket{X_AY_BY_C} + \braket{Y_AX_BY_C} + \braket{Y_AY_BX_C}).
\end{align}
Thus, 
\begin{align}\label{eq:ghzpghzmbound}
F(\rho,\sigma) \leq F(\rho,\sigma) + F(\rho,\ketbra{\overline{\GHZ}_3}) = \frac14(1+3z_2).
\end{align}

Let us now consider the constituents of $z_1$. W.l.o.g., we can assume that $\braket{Z_A}\geq 0$ and $\braket{Z_B}\geq 0$, because otherwise we could consider a permutation of the state $X_AX_BX_C\rho X_AX_BX_C$, which can be prepared in the triangle network without shared randomness as well, has the same fidelity with the GHZ state and negated (and possibly permuted) local $Z$-correlations. As only $\braket{Z_C}$ might be negative, we have that $z_1\in[-1/3, 1]$. Let us consider the two cases of positive or negative $\braket{Z_C}$ individually.
\begin{description}
    \item[Case 1, $\braket{Z_C} < 0$]
        In this case, we know from Eq.~(44) in the Supplemental Material of Ref.~\cite{gisin2020constraints} that the following inequality holds for triangle network states without shared randomness:
        \begin{multline}\label{eq:gisin1improved}
            (1+\braket{Z_A}+\braket{Z_B}+\braket{Z_AZ_B})^2 + (1+\braket{Z_A}-\braket{Z_C}+\braket{Z_AZ_C})^2 + (1+\braket{Z_B}-\braket{Z_C}+\braket{Z_BZ_C})^2 \\
             \leq 6(1+\braket{Z_A})(1+\braket{Z_B})(1-\braket{Z_C})+8\braket{Z_C}(\braket{Z_A} + \braket{Z_B}+\braket{Z_A}\braket{Z_B}). 
        \end{multline}
        Furthermore, positivity demands
        \begin{align}
            &1 - \braket{Z_A} - \braket{Z_B} + \braket{Z_C} + \braket{Z_AZ_B} - \braket{Z_AZ_C} - \braket{Z_BZ_C} + \braket{Z_AZ_BZ_C} \geq 0 \\
            \Rightarrow\quad &1 - \braket{Z_A} - \braket{Z_B} + \braket{Z_C} + \braket{Z_AZ_B} - \braket{Z_AZ_C} - \braket{Z_BZ_C} + 1 \geq 0, \label{eq:ineq_positivity110}
        \end{align}
        as the first line corresponds to the non-negative probability of obtaining outcomes $1,1,0$ when measuring in the computational basis. The second line follows from bounding the tripartite correlation $\braket{Z_AZ_BZ_C}$, by its maximal value of 1.
        
        We can now consider the constrained optimization problem
        \opti{z_2^* =~}{max}{\braket{Z_i},\braket{Z_{k}Z_{l}}}{z_2}{\phantom{-}0\leq \braket{Z_A} \leq 1,; \phantom{-}0\leq \braket{Z_B} \leq 1,; -1 \leq \braket{Z_C} < 0,;\text{Eq.~\eqref{eq:gisin1improved} holds},;\text{Eq.~\eqref{eq:ineq_positivity110} holds}. }
        
        Let us simplify this optimization problem by showing that $\braket{Z_A} = \braket{Z_B}$ and $\braket{Z_AZ_C} = \braket{Z_BZ_C}$ can be assumed at the maximum. We show this by demonstrating that, given a solution to the optimization problem, symmetrizing these expectations values leads to another valid solution.
        To see this, we restrict ourselves to an optimization on the planes corresponding to constant sums $\braket{Z_A}+\braket{Z_B}$ and $\braket{Z_AZ_C}+\braket{Z_BZ_C}$ by setting $\braket{Z_B} = \bar{z}_1 - \braket{Z_A}$ and $\braket{Z_BZ_C} = \bar{z}_2 - \braket{Z_AZ_C}$ for fixed $\bar{z}_1$, $\bar{z}_2$. 
        Note that the optimization function $z_2$, as well as the left-hand side of Eq.~\eqref{eq:ineq_positivity110} are independent of the single expectation values and only depend on the constant sums $\bar{z}_1$ and $\bar{z}_2$, whereas Eq.~\eqref{eq:gisin1improved} becomes
        \begin{multline}
            (1+\bar{z}_1+\braket{Z_AZ_B})^2 + (1+\braket{Z_A} -\braket{Z_C} + \braket{Z_AZ_C})^2 + (1+\bar{z}_1 + \bar{z}_2  - \braket{Z_A} - \braket{Z_C} - \braket{Z_AZ_C})^2 \\
            - 6(1+\braket{Z_A})(1+\bar{z}_1 -  \braket{Z_A})(1-\braket{Z_C}) - 8\braket{Z_C}(\bar{z}_1 + (\bar{z}_1 - \braket{Z_A})\braket{Z_A}) \leq 0.
        \end{multline}
        While this might not look appealing, it becomes least restrictive, i.e., the left-hand side becomes minimal, for the symmetric case. To see this, consider the Hesse matrix w.r.t.~the variables $\braket{Z_A}$ and $\braket{Z_AZ_C}$, the eigenvalues of which read
        \begin{align}
            \lambda_{\pm} = 2 \left(5+\braket{Z_C} \pm \sqrt{\braket{Z_C}^2+6 \braket{Z_C}+13}\right),
        \end{align}
        which are strictly positive for $-1\leq \braket{Z_C} < 0$. Thus, the left-hand side is convex w.r.t.~these two variables. Furthermore, its partial derivatives w.r.t.~the two variables vanish at $\braket{Z_A} = \bar{z}_1/2$, $\braket{Z_AZ_C} = \bar{z}_2 / 2$, showing that the global minimum of the left-hand side is assumed for the choice of $\braket{Z_A} = \braket{Z_B}$ and $\braket{Z_AZ_C} = \braket{Z_BZ_C}$. Inserting this into the optimization problem allows us to analytically solve it using the computer algebra system Mathematica and obtain
        $z_2^* =\left(42-5 \sqrt{42}\right)/21$. Inserting this into the fidelity bound of Eq.~\eqref{eq:ghzpghzmbound} yields
        \begin{align}
            F \leq \frac{1}{28} \left(49-5 \sqrt{42}\right) \approx 0.593.
        \end{align}
        As we will see below, we will obtain a larger bound in the other case, such that we can dismiss this one.

    \item[Case 2, $\braket{Z_C} \geq 0$]
        In this case, we use two inequalities valid for triangle network states without shared randomness: First, we use Gisin's inequality from Ref.~\cite{gisin2020constraints} for the case of positive local correlators, which reads
        \begin{multline}\label{eq:gisin1}
            (1+\braket{Z_A}+\braket{Z_B}+\braket{Z_AZ_B})^2 + (1+\braket{Z_A}+\braket{Z_C}+\braket{Z_AZ_C})^2 + (1+\braket{Z_B}+\braket{Z_C}+\braket{Z_BZ_C})^2 \\
                 \leq 6(1+\braket{Z_A})(1+\braket{Z_B})(1+\braket{Z_C})
        \end{multline}
        Second, we make use of a variant of the classical Finner inequality, established in Ref.~\cite{renou1}.\footnote{Note, however, that the proof in that reference is not complete.} It states that for each local projective measurement $\{P_i\}_{i=0}^{d-1}$ and all triangle network states without shared randomness, 
        \begin{align} \label{eq:finner}
            \trace(\rho P_i\otimes P_j \otimes P_k) \leq \sqrt{\trace(\rho P_i\otimes \one \otimes \one)\trace(\rho \one \otimes P_j \otimes \one)\trace(\rho \one \otimes \one \otimes P_k) }
        \end{align}
        holds. If we choose the dichotomic measurement in the $Z$ basis, i.e., $P_0 = \ketbra{0}$, $P_1 = \ketbra{1}$, and consider $i=j=k=1$, this yields similar to the step from Eq.~\eqref{eq:fidelity_proof2} to Eq.~\eqref{eq:Fbound_116},
        \begin{align} \label{eq:finner2}
            1 - 3z_1 + 3z_2 - z_3 \leq \sqrt{8}\sqrt{(1-\braket{Z_A})(1-\braket{Z_B})(1-\braket{Z_C})}.
        \end{align}
        
        We now maximize the right-hand side of Eq.~\eqref{eq:Fbound_116}, i.e., 
        \opti{F \leq ~}{max}{\braket{Z_i},\braket{Z_{k}Z_{l}},\braket{Z_AZ_BZ_C}}{\frac{1}{16}\Big(\sqrt{1+3z_1+3z_2+z_3}+\sqrt{1-3z_1+3z_2-z_3}\Big)^2}{0\leq \braket{Z_A},\braket{Z_B},\braket{Z_C} \leq 1,;\text{Eq.~\eqref{eq:gisin1} holds},;\text{Eq.~\eqref{eq:finner2} holds}. }

        Let us argue again that in order to maximize the expression, we can assume symmetric choices $\braket{Z_A} = \braket{Z_B} = \braket{Z_C}$ as well as $\braket{Z_AZ_B} = \braket{Z_AZ_C} = \braket{Z_BZ_C}$. To that end, we consider the planes of fixed $z_1$ and $z_2$ and set $\braket{Z_C} = 3z_1 - \braket{Z_A} - \braket{Z_B}$, $\braket{Z_AZ_B} = 3z_2 -\braket{Z_AZ_C} - \braket{Z_BZ_C}$.
        Note that the objective function is independent of the single expectation values and only depends on the constant sums $z_1$ and $z_2$.
        While the left-hand side of Inequality~\eqref{eq:finner2} remains independent as well, its right-hand side is clearly maximized (and therefore least restrictive) for the symmetric choice. However, the inequality~\eqref{eq:gisin1} becomes
        \begin{multline}\label{eq:gisin2plane}
            (1 + 3 z_2+\braket{Z_A}+\braket{Z_B}-\braket{Z_AZ_C}-\braket{Z_BZ_C})^2 +(1+ 3 z_1-\braket{Z_B}+\braket{Z_AZ_C})^2 + (1+ 3 z_1-\braket{Z_A}+\braket{Z_BZ_C})^2 \\ -6 (1+ \braket{Z_A}) (1+ \braket{Z_B}) (1 + 3 z_1 -\braket{Z_A}-\braket{Z_B}) \leq 0.
        \end{multline}
        We again minimize the left-hand side of this inequality to obtain the least restrictive one. To that end, we calculate its $4\times 4$-dimensional Hesse matrix w.r.t.~the variables $\braket{Z_A}, \braket{Z_B}, \braket{Z_AZ_C}, \braket{Z_BZ_C}$. Its characteristic polynomial reads
        \begin{align}
            x^4 + a_3 x^3 + a_2 x^2 + a_1 x + a_0,
        \end{align}
        where
        \begin{align*}
            a_0 &= 432[&-9 z_1^2+6 z_1+3&&+12 z_1 (\braket{Z_A}+ \braket{Z_B})&-\phantom{0}4( \braket{Z_A}^2+\braket{Z_A} \braket{Z_B}+ \braket{Z_B}^2)&],\\
            a_1 &= 288[&9 z_1^2-6 z_1-5&&-(12 z_1 + 1) (\braket{Z_A}+ \braket{Z_B})&+\phantom{0}4 (\braket{Z_A}^2+ \braket{Z_A} \braket{Z_B}+ \braket{Z_B}^2)& ],\\
            a_2 &= \phantom{0}12[&-27 z_1^2+24 z_1+35&& + (36 z_1+8) (\braket{Z_A}+\braket{Z_B})&-12( \braket{Z_A}^2+\braket{Z_A} \braket{Z_B}+ \braket{Z_B}^2)&],\\
            a_3 &= &-40&& -12 (\braket{Z_A}+ \braket{Z_B})&&.\\
        \end{align*}
        We can trivially maximize the prefactors of $x^3$ and $x$, as well as minimize those of $x^2$ and $x^0$ over $\braket{Z_A}$, $\braket{Z_B}$ and $z_1$ under the constraints $0 \leq \braket{Z_A}, \braket{Z_B} \leq 1$ and $\braket{Z_A} + \braket{Z_B} \leq  3z_1 \leq  3$ using the computer algebra system Mathematica to show that the signs of the prefactors alternate. Thus, by Descartes' rule of signs, the Hesse matrix is positive and therefore the left-hand side of \eqref{eq:gisin2plane} is convex. Setting the first derivatives equal zero yields a single solution in the optimization region, namely $\braket{Z_A} = \braket{Z_B} = \braket{Z_C}$ and $\braket{Z_AZ_B} = \braket{Z_AZ_C} = \braket{Z_BZ_C}$, which is therefore the global minimum and thus the least restrictive choice for the optimization.

        Finally, we can thus consider the symmetrized optimization problem where each occurrence of $\braket{Z_A}$, $\braket{Z_B}$ and $\braket{Z_C}$ is replaced by $z_1$ (and likewise $\braket{Z_AZ_B} = \ldots = z_2$), yielding
        \opti{F \leq ~}{max}{z_1,z_2,z_3}{\frac{1}{16}\Big(\sqrt{1+3z_1+3z_2+z_3}+\sqrt{1-3z_1+3z_2-z_3}\Big)^2}{0\leq z_1 \leq 1,;3(1+2z_1+z_2)^2 \leq 6(1+z_1)^3;0\leq 1-3z_1+3z_2-z_3 \leq \sqrt{8}\sqrt{(1-z_1)^3}. }
        This optimization can be solved using the computer algebra system Mathematica and yields a maximum in terms of the largest real root of the polynomial
        \begin{align}
            832000 x^6 + 114944 x^5 - 
   2964416 x^4+ 4210624 x^3 - 2743640 x^2 + 889800 x - 120125,
        \end{align}
        which is approximately given by 
        \begin{align}
            F\leq 0.618,
        \end{align}
        where we rounded up the last digit.
        This proves the claim.
\end{description}

The derivation of the analytical bound in Result~\ref{res:GHZ32bound} makes use of a couple of inequalities and estimates which raises the question of how wasteful they are.  To answer this question, we resort to convex optimization. However, both the Gisin inequality in Eq.~\eqref{eq:gisin1} as well as the Finner inequality in Eq.~\eqref{eq:finner2} are not convex. Thus, we apply the following trick: Instead of optimizing the fidelity over general output states, we fix the values of $\langle Z_A \rangle \equiv z_A$ and $\langle Z_B \rangle \equiv z_B$ and treat $z_A$ and $z_B$ as parameters between $0$ and $1$. In order to get rid of the absolute values around $\langle Z_C \rangle$, we treat its sign as an additional dichotomic input parameter $c\in \{-1,1\}$. When maximizing the inequality, we then consider individually the cases of $\langle Z_C \rangle \geq 0$ (i.e., $c=1$) and $\langle Z_C \rangle < 0$ (i.e., $c=-1$) and take the maximum of both optimizations. We therefore consider the parameterized optimization problem
\begin{figure}[t]
    \centering
    \includegraphics[width=0.7\linewidth]{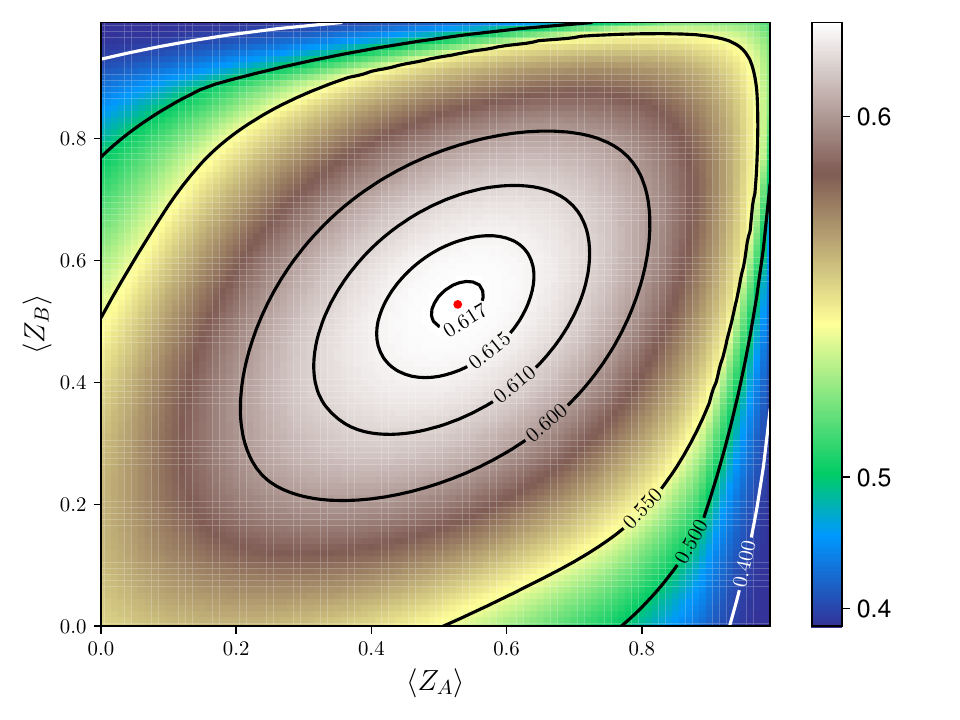}
    \caption{The maximal GHZ-fidelities in triangle LOSR networks $B(z_A,z_B)$ from Eq.~\eqref{eq:optival2} for fixed values of $z_{A} = \langle Z_{A}\rangle$ and $z_{B} = \langle Z_{B}\rangle$. The optimization takes into account Gisin's inequality \eqref{eq:gisin1} and the Finner inequality \eqref{eq:finner2}. The image is obtained by rasterization of the parameters in steps of $10^{-2}$. The global optimum (marked with a red dot) reaches the bound from Result~\ref{res:GHZ32bound}, indicating that no further improvement of the upper bound is possible if only Gisin's inequality and the Finner inequality are taken into account.}
    \label{fig:optiplot}
\end{figure}

\opti{B(z_A,z_B,c) =~}{max}{\rho}{\braket{\GHZ_3|\rho|\GHZ_3} \label{eq:optival1}}{\rho\geq 0,;\trace(\rho)=1,;\langle Z_A\rangle = z_A,;\langle Z_B\rangle = z_B,;c\langle Z_C\rangle \geq 0,;\text{Eq.~\eqref{eq:gisin1} (if $c=1$) or Eq.~\eqref{eq:gisin1improved} (if $c=-1$) holds,};\text{Eq.~\eqref{eq:finner2} holds.}}

We can turn this optimization into the form of a semidefinite program by noticing that for $c=1$, the constraint \eqref{eq:gisin1} can be cast into the form of a positivity constraint of a symmetric matrix, namely
\begin{align}
    \begin{pmatrix}
        6(1+z_A)(1+z_B)(1+\braket{Z_C}) & \ldots &  \ldots & \ldots \\
        1+z_A+z_B+\braket{Z_AZ_B} & 1 & \ldots & \ldots\\
        1+z_A+\braket{Z_C}+\braket{Z_AZ_C} & 0 & 1 & \ldots\\
        1+z_B+\braket{Z_C} + \braket{Z_BZ_C} & 0 & 0 & 1
    \end{pmatrix}\geq 0
\end{align}
(the entries filled with $\ldots$ are fixed by choosing the matrix symmetric).
Likewise, if $c=-1$, \eqref{eq:gisin1improved} can be written as positivity of the matrix
\begin{align}
    \begin{pmatrix}
        6(1+z_A)(1+z_B)(1-\braket{Z_C}) + 8\braket{Z_C}(z_A+z_B+z_Az_B) & \ldots & \ldots & \ldots \\
        1+z_A+z_B+\braket{Z_AZ_B} & 1 & \ldots & \ldots\\
        1+z_A-\braket{Z_C}+\braket{Z_AZ_C} & 0 & 1 & \ldots\\
        1+z_B-\braket{Z_C} + \braket{Z_BZ_C} & 0 & 0 & 1
    \end{pmatrix}\geq 0.
\end{align}
Finally, the Finner inequality from Eq.~\eqref{eq:finner2} can be written as
\begin{align}
    \begin{pmatrix}
        (1-z_A)(1-z_B)(1-\braket{Z_C}) & \ldots \\
        1-z_A-z_B-\braket{Z_C} + \braket{Z_AZ_B} + \braket{Z_AZ_C} + \braket{Z_BZ_C} - \braket{Z_AZ_BZ_C} & 8
    \end{pmatrix}\geq 0
\end{align}
The virtue of formulating a semidefinite program lies in the fact that we can obtain robust numerical solutions with guaranteed optimality for each choice of $z_A, z_B, c$. We plot the resulting function
\begin{align} \label{eq:optival2}
    B(z_A, z_B) := \max_{c\in\{-1,1\}} B(z_A,z_B,c)
\end{align}
in Fig.~\ref{fig:optiplot} by sampling the space of $z_A, z_B \in [0,1]$ in steps of $10^{-2}$.  
It is apparent from the figure that the maximum occurs at $z_A = z_B \approx 0.528$ and reaches indeed the fidelity of $F \approx 0.618$, implying that the result from Result~\ref{res:GHZ32bound} cannot be improved, if only Gisin's inequality and the Finner inequality are considered (we also added the Finner inequality for different choices of outcomes $i,j,k$, but it did not improve the result). Additionally, we combined our SDP method with the SDP inflation technique that was used in Ref.~\cite{navascues2020genuine}. However, it did not improve the result.

Note that the numerical result also allows us to explicitly state the putative state which reaches this fidelity, namely
\begin{align}\label{eq:optimalstate}
    \rho_{\text{opt}} = \begin{pmatrix} 0.597 & 0 &    0 &     0 &    0 &     0 &    0 &     0.262\\
  0 &     0.071 & 0 &     0 &    0 &     0 &    0 &     0 &   \\
 0 &     0 &     0.071 &0 &    0 &     0 &     0 &    0 &   \\
  0 &     0 &    0 &     0.025&  0 &     0 &    0 &     0 &   \\
 0 &    0 &    0 &     0 &     0.071& 0 &    0 &    0 &   \\
  0 &     0 &     0 &     0 &    0 &     0.025&  0 &     0 &   \\
 0 &    0 &     0 &    0 &    0 &     0 &     0.025& 0 &   \\
  0.262&  0 &    0 &     0 &    0 &     0 &    0 &     0.115
 \end{pmatrix}.
\end{align}
By design, it neither violates Gisin's inequality nor the Finner inequality, however, given the fact that its fidelity deviates from the fidelity reached by our constructions in Appendix~\ref{app:constructions} by about $0.07$, we assume that it is not preparable in an LOSR network. One would need tools beyond Gisin's inequality, the Finner inequality and the SDP inflation technique from Ref.~\cite{navascues2020genuine} to exclude it.

\section{Fidelity bounds with the GHZ state for arbitrary dimension and number of parties}
\label{app:fidbound_ghzn}

In this section we prove the statements of Results~\ref{res:GHZ3dbound} and \ref{res:multipartitebound} from the main text by deriving upper bounds for the fidelity of quantum network states with the GHZ state for any dimension and number of parties, i.e., $F_{\LOSR}(\ket{\GHZ_{N,d}})$.

First, for $N=3$, we use Eq.~\eqref{eq:fidmeasurement} to bound the fidelity of a state $\rho$ with the GHZ state in terms of joint probability distributions by choosing as POVM the projective measurements in the computational basis, i.e. $\{\ketbra{ijk}\}_{i,j,k \in \{0,d-1 \}}$, which yields
\begin{align}
    F(\rho,\ketbra{\GHZ_{3,d}})\le \frac{1}{d}\left(\sum_{i=0}^{d-1}\sqrt{\bra{iii}\rho\ket{iii}} \right)^2= \frac{1}{d}\left(\sum_{i=0}^{d-1}\sqrt{p(i,i,i)} \right)^2.
\end{align}
We make use of Eq.~\eqref{eq:finner}, the quantum variant of the classical Finner inequality for quantum network states without shared randomness \cite{renou1}. Note again that due to the fact that the fidelity with a pure state is a linear function, the derived fidelity bound will also hold for proper LOSR states with shared randomness.
Using for the Finner inequality the same measurement in the computational basis, i.e. $P_i = \ketbra{i}$, we get
\begin{align}
\label{eq:renou1}
    p(i,j,k)\le\sqrt{p_A(i)p_B(j)p_C(k)}
\end{align}
where $p(i,j,k) = \trace(\rho \ketbra{ijk})$, $p_A(i) = \trace(\rho \ketbra{i}\otimes \one \otimes \one)$, $p_B(j) = \trace(\rho \one \otimes \ketbra{j}\otimes \one)$ and $p_C(k) = \trace(\rho \one \otimes \one \otimes \ketbra{k})$ for $i$, $j$, $k \in \{0,1,\ldots,d-1\}$.

Thus, for any state $\rho$ from a triangle network without shared randomness, we have that
\begin{align}
    F(\rho,\ketbra{\GHZ_{3,d}})&\le \frac{1}{d}\left( \sum_{i=0}^{d-1}\left( p_A(i)p_B(i)p_C(i) \right)^{1/4}\right)^2\\
    &=\frac{1}{d}\left( d^{1/4}\sum_{i=0}^{d-1}\left( p_A(i)p_B(i)p_C(i)\frac{1}{d} \right)^{1/4} \right)^2\\
    &\le \frac{1}{\sqrt{d}},
\end{align}
where in the last step we used that the Matusita fidelity is upper bounded by $1$ (see, e.g.,~\cite{wilde1}). The Matusita fidelity is defined as follows: For $r\in \mathbb{N}$, let $\{p_1, p_2, \dots, p_r\}$ be a set of probability distributions on a finite set $\mathcal{X}$, then  $\mathcal{F}_r(p_1, p_2, \dots, p_r)\coloneq \sum_{x \in \mathcal{X}}(p_1(x)\dots p_r(x))^{1/r}$ is called their Matusita fidelity. Note that here we have $r=4$, and the fourth probability distribution $p_4(i) = 1/d$ is the uniform distribution on the set $\{0,\dots,d-1\}$.

Next we consider network states of four parties with bipartite source states without shared randomness. Note that the constraint in Eq.~\eqref{eq:renou1} is extendible to an arbitrary number of parties \cite{renou1}:
\begin{align}
    \label{eq:renou2}
    p(a_1,a_2, \dots a_N)\le \sqrt{p_{A_1}(a_1)p_{A_2}(a_2)\dots p_{A_N}(a_N)}.
\end{align}
Thus, 
\begin{align}
    F(\rho,\ketbra{\GHZ_{4,d}} )&\le \frac{1}{d}\left(\sum_{i=0}^{d-1}\sqrt{\bra{iiii}\rho \ket{iiii}} \right)^2= \frac{1}{d}\left(\sum_{i=0}^{d-1}\sqrt{p(i,i,i,i)} \right)^2\\
    &\le \frac{1}{d}\left(\sum_{i=0}^{d-1}(p_A(i)p_B(i)p_C(i)p_D(i))^{1/4} \right)^2\\
    &=\frac{1}{d}\mathcal{F}_4(p_A, p_B, p_C, p_D)^2\le \frac{1}{d}.
\end{align}
In the first line, we again estimate the fidelity by choosing the projective measurement in the computational basis. In the second line, we use Eq.~\eqref{eq:renou2} and in the last step we use the property that the Matusita fidelity cannot be greater than 1.
Together with Corollary~\ref{cor:ghzfidelity}, it follows that this bound also holds for quantum network states with $N\ge 4$ parties and bipartite source states, i.e.,
\begin{align}
    F(\rho,\ketbra{\GHZ_{N,d}})\le \frac{1}{d}
\end{align}
for all $N\ge 4$.
We also remark that this is a tight bound which can be achieved by the product state $\rho=\ketbra{i}^{\otimes N}$ with $i \in \{0, \dots, d-1\}$.

\end{document}